\documentclass{article}

\pdfoutput=1

\usepackage{amssymb,amsmath,commath,natbib,graphicx,subfigure,amsfonts,amstext,amsthm}
\usepackage{latexsym,indentfirst,enumerate}
\usepackage{xr,color}

\newcommand{\ghadd}[1]{\textcolor{black}{#1}}

\newtheorem{lem}{Lemma}

\newtheorem{thm}{Theorem}
\newtheorem{cor}{Corollary}

\def\real{\mathop{\mathbb R}}      

\def\argmin{\mathop{\rm{argmin}}}

\def\argmax{\mathop{\rm{argmax}}}

\newcommand{\be}{\begin{eqnarray*}}
\newcommand{\ee}{\end{eqnarray*}}

\newcommand{\de}{\delta}

\newcommand{\Ta}{\Theta}
\newcommand{\ta}{\theta}

\newcommand{\beq}{\begin{eqnarray}}
\newcommand{\eeq}{\end{eqnarray}}

\newcommand{\convd}{\stackrel{d}{\rightarrow}}

\newcommand{\convas}{\stackrel{a.s.}{\longrightarrow}}
\newcommand{\betabold}    {\mbox{\boldmath${\beta}$}}

\def\argmax{\operatornamewithlimits{arg\,max}}
\def\argmin{\operatornamewithlimits{arg\,min}}

\title{Bayesian Model Robustness via Disparities}
\author{Giles Hooker and Anand Vidyashankar}

\date{}

\begin{document}

\maketitle

%

\begin{abstract}
This paper develops a methodology for robust Bayesian inference through the use
of disparities.  Metrics such as Hellinger distance and negative exponential
disparity have a long history in robust estimation in frequentist inference. We
demonstrate that an equivalent robustification may be made in Bayesian
inference by substituting an appropriately scaled disparity for the log
likelihood to which standard Monte Carlo Markov Chain methods may be applied. A
particularly appealing property of minimum-disparity methods is that while they
yield robustness with a breakdown point of 1/2,
the resulting parameter estimates are also efficient when the
posited probabilistic model is correct. We demonstrate that a similar property
holds for disparity-based Bayesian inference.  We further show that in the
Bayesian setting, it is also possible to extend these methods to robustify
regression models, random effects distributions and other hierarchical models.
The methods are demonstrated on real world data.
\end{abstract}

\section{Introduction}

In this paper we develop a new methodology for providing robust inference in a Bayesian
context.  When the data at hand are suspected of being contaminated with large
outliers it is standard practice to account for these either (1) by postulating a
heavy-tailed distribution, (2) by viewing the data as a mixture,
with the contamination explicitly occurring as a mixture component
 or (3) by employing priors that penalize large values of a
parameter \citep[see][]{Berger94,Albert09,AndradeOhagan06}. In the context of frequentist inference, these issues are investigated using methods such as M-estimation, R-estimation etc.  and are part of  standard robustness literature \citep[see][]{Hampel86, Maronna06,Jureckova06}.  As is the case for Huberized
loss functions in frequentist
inference, even though these approaches provide robustness they lead to a loss
of precision when contamination is not present in (1) and (2) above or to a distortion of prior
knowledge in (3). \ghadd{Explicit modeling of outliers as in (2) also requires knoweldge of outlier configurations -- how many mixture components to use and what distributions to use them in, for example -- and may not be robust if these are incorrect.}
This paper develops an alternative systematic Bayesian approach, based on disparity theory,
that is shown to provide robust inference without
loss of efficiency for large samples.

In parametric frequentist inference using independent and identically distributed (i.i.d.) data, several authors
\citep{Beran77,TamuraBoos86,Simpson87,Simpson89,ChenVidyashankar06}
have demonstrated that the dual goal of efficiency and robustness is achievable
by using the minimum Hellinger distance estimator (MHDE). In the i.i.d. context,
MHDE estimators are defined by minimizing the Hellinger distance between a
postulated parametric density $f_\ta(\cdot)$ and a non-parametric estimate
$g_n(\cdot)$ over the $p$-dimensional parameter space $\Theta$; that is,
\begin{equation} \label{HD}
\hat{\ta}_{HD} = \arg \inf_{\ta \in \Theta} \int \left(g_n^{1/2}(x) -
f^{1/2}_\ta(x) \right)^2 dx.
\end{equation}
Typically, for continuous data,
$g_n(\cdot)$ is taken to be a kernel density estimate; if the probability model is supported
on discrete values, the empirical distribution is used.  More generally, \citet{Lindsay94}
introduced the concept of a minimum disparity procedure; developing a class of
divergence measures that have similar properties to minimum Hellinger distance estimates.
These have been further developed in \citet{BSV97} and \citet{ParkBasu04}.
\citet{HookerVidyashankar09b} have extended these methods to a regression framework.

A remarkable property of disparity-based estimates is that while they confer
robustness, they are also first-order efficient. That is, they obtain the
information bound when the postulated density $f_{\ta}(\cdot)$ is correct. In
this paper we develop robust Bayesian inference using disparities. We show that
appropriately scaled disparities  approximate $n$ times the negative
log-likelihood near the true parameter values. We use this as a motivation to
replace the log likelihood in Bayes rule  with a disparity to create  what we
refer to as the ``D-posterior''. We demonstrate that this technique is readily
amenable to Markov Chain Monte Carlo (MCMC) estimation methods. Finally, we
establish that the expectation of the D-posterior is asymptotically efficient
and the resulting credible intervals provide asymptotically accurate coverage
when the proposed parametric model is correct.

Disparity-based robustification in Bayesian inference can be naturally extended
to a regression framework through the use of conditional density estimation as
discussed in \citet{HookerVidyashankar09b}. We pursue this extension to
hierarchical models and replace various terms in the hierarchy with
disparities. This creates a novel ``plug-in procedure'' -- allowing the
robustification of inference with respect to particular distributional
assumptions in complex models. We develop this principle and demonstrate its
utility on a number of examples. The use of a disparity within a Bayesian
context imposes an additional computational burden through the estimation of a
kernel density estimate and the need to run MCMC methods. Our analysis and
simulations demonstrate that while the use of MCMC significantly increases
computational costs, the additional cost of the use of disparities is on the
order of a factor between 2 and 10, remaining implementable for many applications.
These methods require marginalization of an exponentiated disparity with respect to
the random effects distribution; a task that can be achieved through MCMC methods, but
would otherwise be numerically challenging.

The use of divergence measures for outlier analysis in a Bayesian context has
been considered in \citet{dey} and \citet{PD}. Most of this work is concerned
with the use of divergence measures to study Bayesian robustness when the
priors are contaminated and to diagnose the effect of outliers. These divergence
measures are computed using MCMC techniques. More recently, \citet{zhanx} and \citet{SRL10} have developed analogues of R-estimates and Bayesian Sandwich estimators. These methods can be viewed to be extensions of robust frequentist methods to Bayesian context. By contrast, our paper is based
on explicitly replacing the likelihood with a disparity in order to provide
a systematic approach  to obtain inherently robust and efficient inference.

Within the context of Bayesian analysis, robustness has been studied with respect to the specification of both prior and data distributions. Robustness to outliers as studied in the frequentist literature is referred to as ``outlier-rejection'' in Bayesian analysis and is studied for example in  \citet{dawid73}, \citet{Ohagan79}, \citet{Ohagan90}, \citet{choy97} and \citet{desgagne07}. Here, outlier rejection indicates that as some group of data is moved to infinity, the posterior reverts to the posterior without those observations. This corresponds to a breakdown point of 1; a rather extreme value for frequentist robustness. We also obtain this breakdown point, but additionally develop a notion of an asymptotic breakdown point in which we examine the worst-case displacement as sample-size increases. We are able to show that this notion effectively describes robustness and distinguishes Bayesian methods along with regularized versions of robust estimators from estimators that are trivially made robust by, for example, threshholding their estimates.

The remainder of the paper is structured as follows: we provide a formal
definition of the disparities in Section \ref{sec:MDE}. Disparity-based
Bayesian inference are developed in Section \ref{sec:dposterior}. Robustness
and efficiency of these estimates are demonstrated theoretically and through a
simulation for i.i.d. data in Section \ref{sec:simple}. The methodology is
extended to regression models in Section \ref{sec:conditional}. The plug-in
procedure is presented in Section \ref{sec:plugin} through an application to a
one-way random-effects model.
Section
\ref{sec:examples} is devoted to two real-world data sets where we apply these
methods to generalized linear mixed models and a random-slope random-intercept
models for longitudinal data.  Proofs of  technical results and details of
simulation studies are relegated to an online appendix.

\section{Disparities and Their Numerical Approximations} \label{sec:MDE}

In this section we describe a class of disparities and numerical
procedures for evaluating them.  These disparities compare a
proposed parametric family of densities to a non-parametric density estimate.
We assume that we have i.i.d. observations $X_i$ for $i = 1,\ldots,n$ from some
density $h(\cdot)$. We let $g_n$ be the kernel density estimate:
\begin{equation} \label{iid.gn}
g_n(x) = \frac{1}{n c_n} \sum_{i=1}^n K \left( \frac{x-X_i}{c_n} \right)
\end{equation}
where the kernel $K$ density and $c_n$ is a bandwidth for the kernel. If $c_n
\rightarrow 0$ and $n c_n \rightarrow \infty$ it is known that $g_n(\cdot)$ is
an $L_1$-consistent estimator of $h(\cdot)$ \citep{DevroyeGyorfi85}. In
practice, a number of plug-in bandwidth choices are available for $c_n$
\citep[e.g.][]{Silverman82, SheatherJones91, EJS94}. \ghadd{For non-i.i.d. data examined in Sections \ref{sec:conditional} and \ref{sec:plugin},  plug-in bandwidths can be calculated from method of moments estimates. We have found our results to be insensitive to the choice of plug-in bandwidth selector.}

We begin by reviewing the class of disparities described in \citet{Lindsay94}.
The definition of disparities involves the residual function,
\begin{eqnarray}
\delta_{\ta,g}(x)= \frac{g(x)- f_\ta(x)}{f_\ta(x)},
\end{eqnarray}
defined on the support of $f_\ta(x)$ and a function $G:[-1,\ \infty) \rightarrow \mathcal{R}$. $G(\cdot)$ is assumed to be
strictly convex and thrice
differentiable with $G(0) = 0$, $G'(0) = 0$ and $G''(0) = 1$.  The disparity between $f_\ta$ and $g_n$ is
defined to be
\begin{eqnarray}{\label{DEEB0}}
D(g_n, f_\ta) = \int_{{\cal R}}G(\de_{\ta, g_n}(x))f_\ta(x) dx.
\end{eqnarray}
An estimate of $\ta$ obtained by minimizing (\ref{DEEB0}) is called
a {\em minimum disparity estimator}. Under differentiability assumptions, this is equivalent to solving the equation
\[
\nabla_\ta D(g_n,f_\ta) = \int A(\delta_\ta(x)) \nabla_\ta f_\ta(x) dx = 0,
\]
where $A(\delta) = G(\delta) - (1+\delta) G'(\delta)$ and $\nabla_\ta$
indicates the derivative with respect to $\ta$.

This framework contains
Kullback-Leibler divergence as approximation to the likelihood:
\[
KL(g_n,f_\ta) = -\int \left( \log f_\ta(x) \right) g_n(x) dx \approx
-\frac{1}{n} \sum_{i=1}^n \log f_\ta(x_i)
\]
for the choice $G(\delta) = (\delta+1)\log(\delta+1)$ \ghadd{up to a constant}.
The squared Hellinger disparity (HD) corresponds to the
choice $G(x)=[(x+1)^{1/2}-1]^2-1$. While robust statistics is typically concerned
with the impact of outliers, the alternate problem of {\em inliers} -- defined as
nominally-dense regions that lack empirical data and consequently small values
of $\de_{\ta, g_n}(x)$ -- can also cause instability. It has been illustrated
in the literature that HD down weighs the effect of large values of $\de_{\ta,
g_n}(x)$ (outliers) relative to the likelihood but magnifies the effect of
inliers. An alternative, the negative exponential disparity, based on the
choice $G(x)= e^{-x}-1$ down weighs the effect of both outliers and inliers.

The integrals involved in (\ref{DEEB0}) are not analytically
tractable and the use of Monte Carlo integration to approximate the objective function has been
suggested in \citet{ChenVidyashankar06}. More specifically, if $z_1,\ldots,z_N$ are i.i.d.
random samples generated from $g_n(\cdot)$, one can approximate $D(g_n, f_\ta)$ by
\begin{equation} \label{mcestimate}
\hat{D}(g_n, f_\ta) = \frac{1}{N} \sum_{i=1}^N G(\de_{\ta,
g_n}(z_i))\frac{f_\ta(z_i)}{g_n(z_i)}.
\end{equation}
The $z_i$ can be efficiently generated in the form $z_i = c_n W_i + X_{N_i}$ for
$W_i$ a random variable generated according to $K$ and $N_i$ sampled uniformly from
the integers $1,\ldots,N$. In the specific case of Hellinger distance
approximation, the above reduces to
\[
\widehat{HD}^2(g_n,f_\ta) = 2 - \frac{2}{N} \sum_{i=1}^N
\frac{f^{1/2}_\ta(z_i)}{g_n^{1/2}(z_i)}.
\]
The use of a fixed set of Monte Carlo samples from $g_n(\cdot)$ when optimizing
for $\ta$ provides a stochastic approximation to an objective function that
remains a smooth function of $\ta$ and hence avoids the need for complex
stochastic optimization. Similarly, in the present paper, we hold the $z_i$
constant when applying MCMC methods to generate samples from the posterior
distribution in order to improve their mixing properties. If $f_\ta$ is
Gaussian with $\ta = (\mu,\sigma)$, Gauss-Hermite quadrature rules can be used
to avoid Monte Carlo integration, leading to improved computational efficiency
in some circumstances. In this case we have
\begin{equation} \label{GHestimate}
\tilde{D}(g_n,f_\ta) = \sum_{i=1}^M w_i(\ta) G(\delta_{\ta,n}(\xi_i(\ta))),
\end{equation}
where the $\xi_i(\ta)$ and $w_i(\ta)$ are the points and weights for a
Gauss-Hermite quadrature scheme for parameters $\ta = (\mu,\sigma)$. The choice
between \eqref{mcestimate} and \eqref{GHestimate} depends on the disparity and
the problem under investigation. When $g_n(\cdot)$ has many local modes,
\eqref{GHestimate} can result in choosing parameters for which some quadrature
point coincides with a local modes. However, \eqref{mcestimate} can be rendered
unstable by the factor $f_\ta(z_i)/g_n(z_i)$ for $\ta$ far from the maximizing value of $D(g_n,f_\ta)$. In
general, we have found \eqref{mcestimate} preferable when using Hellinger
distance, but that \eqref{GHestimate} performs better with negative exponential
disparity.  The relative computational cost of using $\hat{D}(g_n, f_\ta)$ versus
$\tilde{D}(g_n,f_\ta)$ in various circumstances is discussed in Online Appendix
\ref{computation}.

\section{The D-Posterior and MCMC Methods} \label{sec:dposterior}

We begin this section by a heuristic description of the second-order approximation of
$KL(f_\ta,g_n)$ by $D(f_\ta,g_n)$.  A Taylor expansion of
$KL(f_\ta,g_n)$ about $\ta$ has the following first two terms:
\begin{align}
\nabla^2_\ta KL(g_n,f_\ta) & = \int \left[
\frac{1}{f_\ta(x)}\left(\nabla_\ta f_\ta(x)\right)\left(\nabla_\ta
f_\ta(x)\right)^T - \nabla^2_\ta f_\ta(x) \right] (\delta_{\ta,g_n}(x)+1) dx. \label{KLd2} \\
& = \int \left[ \left(\frac{\nabla_\ta
f_\ta(x)}{f_\ta(x)}\right)\left(\frac{\nabla_\ta f_\ta(x)}{f_\ta(x)}\right)^T
-  \frac{\nabla^2_\ta f_\ta(x)}{f_\ta(x)} \right] g_n(x)  dx \nonumber
\end{align}
where the second term approximates the observed Fisher Information when the
bandwidth is small.  The equivalent terms for $D(g_n,f_\ta)$ are:
\begin{align}
\nabla^2_\ta D(g_n,f_\ta) & = \int \nabla^2_\ta f_\ta(x)
A(\delta_{\ta,g_n}(x))dx
\label{Dd1} \\
& \hspace{1cm}
- \int \frac{1}{f_\ta(x)}\left(\nabla_\ta f_\ta(x)\right) \left(\nabla_\ta
f_\ta(x)\right)^T (\delta_{\ta,g_n}(x)+1) A'(\delta_{\ta,g_n}(x)) dx. \nonumber
\end{align}
Now, if $g_n$ is consistent, $\delta_{\ta,g_n}(x) \rightarrow 0$ almost surely (a.s.). Observing
that $A(0)=0$, $A'(0) = -1$ from the conditions on $G$ and observing $\int \nabla^2_\ta f_\ta(x)dx=0$, we obtain the equality
of \eqref{KLd2} and \eqref{Dd1}. The fact that these heuristics yield
efficiency was first noticed by \citet{Beran77} (eq. 1.1).

In the context of Bayesian methods, inference is based on the posterior
\begin{equation} \label{lposterior}
P(\ta|x) = \frac{ P(x|\ta) \pi(\ta)}{ \int P(x|\ta) \pi(\ta) d\ta},
\end{equation}
where  $P(x|\ta) = \exp(\sum_{i=1}^n \log f_\ta(x_i))$ and $\pi$ a prior density which we assume has a first moment.
Following the heuristics above, in this paper
we propose the simple expedient of replacing the log likelihood,
$\log P(x|\ta)$, in \eqref{lposterior} with a disparity:
\begin{equation}
\label{eq:dposterior} P_D(\ta|g_n) = \frac{ e^{-n D(g_n,f_\ta)} \pi(\ta)}{ \int
e^{-n D(g_n,f_\ta)} \pi(\ta) d\ta}.
\end{equation}
In the case of Hellinger distance, the appropriate disparity is
$2HD^2(g_n,f_\ta)$ and we refer to the resulting quantity as the {\em
H-posterior}. When $D(g_n, f_{\ta})$ is based on Negative Exponential disparity, we refer to it as {\em
N-posterior}, and  {\em D-posterior} more generally.  These choices are
illustrated in Figure \ref{likapprox} where we show the approximation of the
log likelihood by Hellinger and negative exponential disparities and the effect
of adding an outlier to these in a simple normal-mean example.

Throughout the examples below, we employ a Metropolis algorithm based on
a symmetric random walk to draw samples from $P_D(\ta|g_n)$. While the cost of
evaluating $D(g_n,f_\ta)$ is greater than the cost of evaluating the
likelihood at each Metropolis step, we have found these algorithms to be
computationally feasible and numerically stable. Furthermore, the burn-in period
for sampling from $P_D(\ta|g_n)$ and the posterior are approximately the same,
although the acceptance rate of the former is around ten percent higher.

After substituting $-n D(g_n,f_\ta)$ for the log likelihood, it will be useful
to define summary statistics of the $D$-posterior in order to demonstrate their
asymptotic properties. Since the $D$-posterior \eqref{eq:dposterior} is a
proper probability distribution, the Expected D-{\em a posteriori} (EDAP)
estimates exist and are given by
\[
\ta^*_n = \int_\Theta \ta P_D(\ta|g_n) d\ta.
\]
and credible intervals for $\ta$ can be based on the quantiles of
$P_D(\ta|g_n)$. These quantities are calculated via Monte Carlo integration
using the output from the Metropolis algorithm.  We similarly define the Maximum
D-{\em a posteriori} (MDAP) estimates by
\[
{\ta}^+_n = \argmax_{\ta \in \Theta} P_D(\ta|g_n).
\]

In the next section we describe the asymptotic properties of EDAP and MDAP
estimators. In particular, we establish the posterior consistency, posterior
asymptotic normality and efficiency of these estimators and their robustness
properties. Differences between  $P_D(\ta,g_n)$ and the posterior do exist and
are described below:
\begin{enumerate}
\item The disparities $D(g_n,f_\ta)$  have
strict upper bounds; in the case of Hellinger distance $0 \leq HD^2(g_n,f_\ta)
\leq 2$, the upper bound for negative exponential disparity is $e$. This implies
that the likelihood part of the D-posterior, $\exp(-n D(g_n,f_\ta))$, is
bounded away from zero. Consequently, a proper prior $\pi(\theta)$ is required
in order to normalize $P_D(\ta|g_n)$. \ghadd{A random $\ta$ from $\pi(\ta)$ must also have finite expectation in order for the EDAP to be defined}.  In particular, uniform priors on unbounded ranges, along
with most reference priors, cannot be employed here.
Further, the tails of $P_D(\ta|g_n)$
are proportional to that of $\pi(\theta)$.  As a consequence, the breakdown point for the EDAP,
as traditionally defined, is 1. Although note that in Section \ref{sec:robustness} we propose an modified definition
of breakdown which is appropriate for regularized and Bayesian estimators under which EDAP has a breakdown
of 1/2.

\ghadd{These results do not affect the
asymptotic behavior of $P_D(\ta|g_n)$ since the lower bounds decrease with
$n$. Modified D-posteriors based on a transformation $m(D(g_n,f_\ta))$ that removes the upper bound can be defined without affecting either the efficiency or robustness of the resulting EDAP estimates. However, appropriate transformations $m$ will depend on the the parametric family $f_\ta$ and are beyond the scope of this paper.}

\item In Bayesian inference for i.i.d. random variables, the log likelihood is a sum of $n$
terms. This implies that if new data $X_{n+1},\ldots,X_{n^*}$ are
obtained, the posterior for the combined data $X_1,\ldots,X_{n^*}$ can be
obtained by using posterior after $n$ observations, $P(\ta|X_1,\ldots,X_n)$ as
a prior $\ta$:
\[
P(\ta|X_1,\ldots,X_{n^*}) \propto
P(X_{n+1},\ldots,X_{n^*}|\ta)P(\ta|X_1,\ldots,X_n).
\]
By contrast, $D(g_n,f_\ta)$ is generally not additive in $g_n$; hence
$P_D(\ta|g_n)$ cannot be factored as above. Extending arguments in \citet{ParkBasu04},
we conjecture that no disparity that is additive in $g_n$ will yield both robust and
efficient posteriors.

\item While we have found that the same Metropolis algorithms can be effectively
used for the D-posterior as would be used for the posterior, it is not possible to use
conjugate priors with disparities. This removes the possibility of using conjugacy
to provide efficient sampling methods within a Gibbs sampler, although these could
be approximated by combining  sampling from a conditional distribution with a rejection
step.
\end{enumerate}
The idea of replacing log likelihood in the posterior with an alternative
criterion occurs in other settings. See \citet{Sollich02}, for example, in
developing Bayesian methods for support vector machines.  However, we replace the
log likelihood with an approximation that is explicitly designed to be both
robust and efficient, rather than as a convenient sampling tool for a
non-probabilistic model.

\begin{figure}
\begin{center}
\begin{tabular}{ccc}
\includegraphics[height=4.4cm]{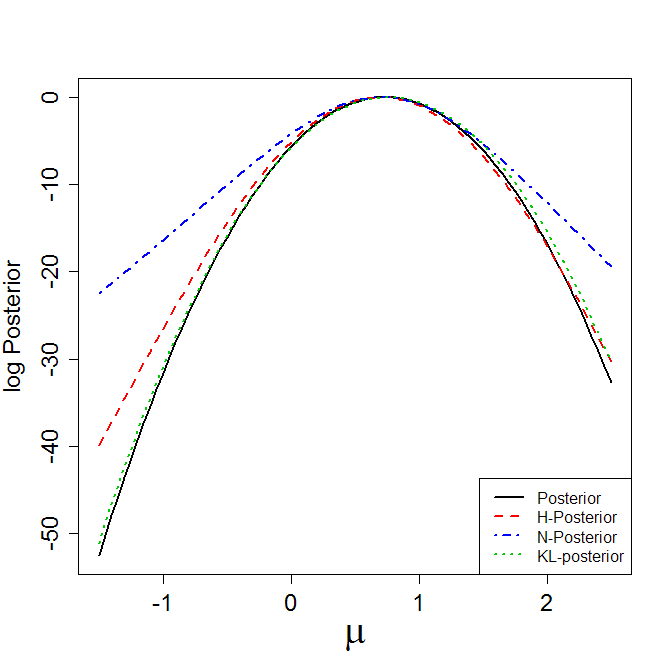} &
\includegraphics[height=4.4cm]{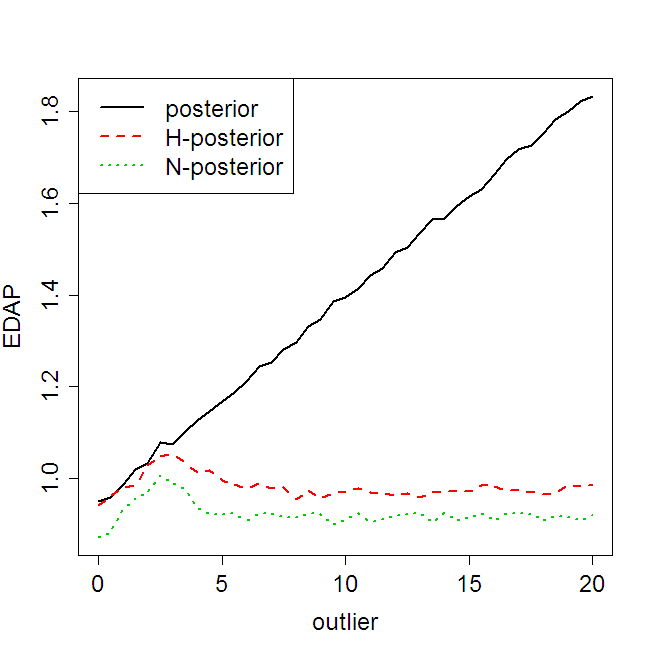} &
\includegraphics[height=4.4cm]{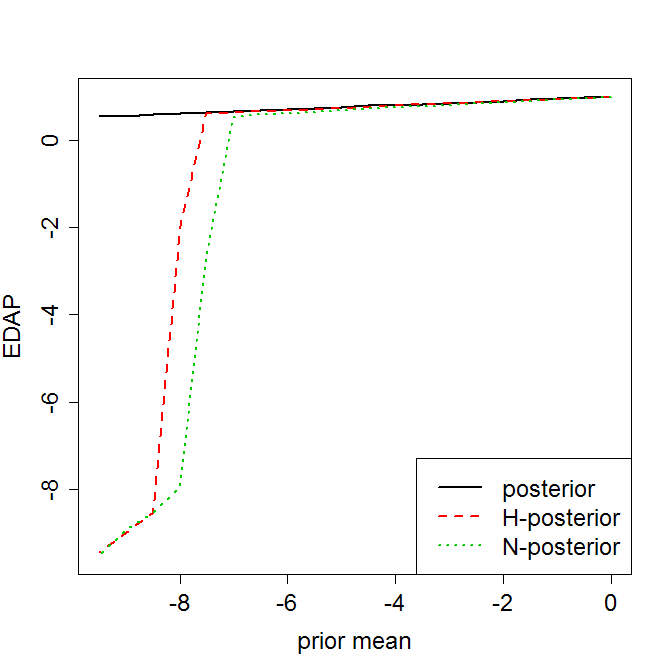}
\end{tabular}
\end{center}
\caption{Left: A comparison of log posteriors for $\mu$ with data generated
from $N(\mu,1)$ with $\mu=1$ using an $N(0,1)$ prior for $\mu$.  Middle:
influence of an outlier on expected {\em D-a posteriori} (EDAP) estimates of
$\mu$ as the value of the outlier is changed from 0 to 20. Right: influence of
the prior as the prior mean is is changed from 0 to -10. } \label{likapprox}
\end{figure}

\section{Robustness} \label{sec:robustness}

The appeal of disparity-based methods is that in addition to the statistical efficiency of the estimators defined above when the parametric model is correctly specified, these estimators are also robust to contamination by data taking large values. As may be expected from the results above, EDAP estimators behave similarly to their minimum-disparity counterparts at finite levels of contamination at large but finite values. However, classical measures of robustness -- influence functions and breakdown points -- are based at limiting values, either of infinitesimal contamination levels or contaminating values at infinity where the convergence between EDAP and minimum-disparity estimators fails. This is due to a lack of uniformity and we argue that the direct application of the robustness measures listed above do not provide an accurate description of the behavior of EDAP estimates. Instead, robustness  should be measured by the properties of the pointwise limit of $\alpha$-level influence functions.

As noted in the introduction, robustness to outliers is treated under title of ``outlier rejection'' in Bayesian analysis and generally corresponds to a breakdown point of 1. As we show below, the analysis of robustness we propose also reconciles Bayesian and other regularized estimates with the traditional description of a robust estimator as having breakdown point of 1/2.  A Bayesian analysis of outlier rejection for our methods can be undertaken using the analysis techniques developed here; it is omitted for the sake of brevity.

To describe robustness, we view our estimates as functionals $T_n(h)$ \ghadd{mapping the space of densities to $\real^p$}. In particular, we examine the EDAP estimate
\begin{equation} \label{edap_functional}
T_n(h) = \frac{\int \ta e^{-nD(h,f_\ta)} \pi(\ta) d\ta}{\int e^{-nD(h,f_\ta)}
\pi(\ta) d\ta}
\end{equation}
and note that in contrast to classical approaches to analyzing robustness the
interaction between the disparity and the prior requires us to make the
dependence of $T_n$ on $n$ explicit. This dependence is shared by any estimator that incorporates priors --  including all classical Bayesian methods -- and affects the traditional measures of robustness as examined below. Note that here, $T_n$ is taken to be a deterministic sequence of maps for the space of densities to $\Theta$. 

We analyze the behavior of $T_n(h)$ under the sequence of perturbations $h_{z,\alpha}(x) = (1-\alpha)g(x)+\alpha t_z(x)$ for a sequence of densities $t_z(\cdot)$ and $0 \leq \alpha \leq 1$. Here we assume that $t_z(\cdot)$ is a contaminating sequence defined so that it becomes orthogonal to both $g$ and the parametric family for large $z$. Note that unlike our examination of efficiency below, in these analyzes we do not require that $g$ belongs to the parametric family; thus $h_{\alpha,z}$ describes the effect of adding outliers to a fixed kernel density estimate. We also assume that $f_\ta$ and $g$ become orthogonal at large values of $\ta$.
\begin{align}
& \lim_{z \rightarrow \infty} \int t_z(x) g(x) dx = 0 \label{tgortho} \\
& \lim_{z \rightarrow \infty} \int t_z(x) f_\ta(x) dx = 0, \ \forall \ta \in \Theta \label{tfortho}\\
& \lim_{\ta^* \rightarrow \infty} \sup_{\|\ta\| > \ta^*} \int g(x) f_{\ta}(x) dx  = 0. \label{fgortho}
\end{align}
Typically, $t_z(\cdot)$ is taken to be a uniform distribution on a small neighborhood centered at $z$;
\ghadd{but these conditions are clearly more general. They extend those given in \citet{ParkBasu04} in not requiring $g$ to be a member of the parametric family $f_\ta$. } The $\alpha$-level influence function is then defined analogously to \citet{Beran77} by
\begin{equation} \label{alpha.influence}
\mbox{IF}_{\alpha,n}(z) = \alpha^{-1} \left[ T_n(h_{z,\alpha}) - T_n(h) \right]
\end{equation}
where we again note that the dependence of $\mbox{IF}_{\alpha,n}(z)$ on $n$ is induced by the prior.

\eqref{alpha.influence} represents a complete description of the behavior of our estimator in the presence of contamination, up to the shape of the contaminating density. While for EDAP estimators it contains an explicit dependence on $n$, we begin by observing its limit for large $n$. Firstly, as with more classical Bayesian estimates, EDAP estimators approach their frequentist counterparts at a $n^{-1}$ rate.
\begin{thm} \label{EDAP-MDE.thm}
Assume that $G$ has four continuous derivatives, that $f_\ta$ is four times
continuously differentiable in $\ta$ and that the third derivatives of $\pi$ are bounded. Define $T_n(h)$ as in \eqref{edap_functional} and the minimum disparity estimator (MDE) as
\begin{equation} \label{MDE}
\hat{\ta}(h) = \argmin_{\ta \in \Theta} D(h,f_\ta)
\end{equation}
then $T_n(h) - \hat{\ta}(h) = o_p(n^{-1})$.
\end{thm}
\noindent As an immediate corollary, the $\alpha$-level influence functions converge at the same rate
\begin{cor} \label{alpha.influence.cor}
Under the conditions of Theorem \ref{EDAP-MDE.thm}, define
\[
\mbox{IF}_{\alpha,\infty}(z) = \alpha^{-1} \left[ \hat{\ta}(h_{z,\alpha})-\hat{\ta}(h) \right]
\]
then for every $\alpha$ and $z$,
\begin{equation} \label{influence.convergence}
\mbox{IF}_{\alpha,n}(z) - \mbox{IF}_{\alpha,\infty}(z) = o_p(n^{-1}).
\end{equation}
\end{cor}
That is, the influence function for $\hat{\theta}(\cdot)$ represents a reasonable description of the behavior
of $T_n(\cdot)$.  In particular,  in the case of Hellinger distance methods, Theorems 5 and 6 in \citet{Beran77} have direct analogues for EDAP and MDAP estimators respectively.

While Corollary \ref{alpha.influence.cor} motivates using the properties of $\mbox{IF}_{\alpha,\infty}(z)$ to describe the robustness of of $T_n(\cdot)$, we note that the convergence in \eqref{influence.convergence} need not be uniform in $z$. The classical summaries of robustness properties investigated below are focussed on extremal values of $\mbox{IF}_{\alpha,n}(z)$; the breakdown point at large $z$ and the classical influence function at small $\alpha$. The lack of uniformity in \eqref{influence.convergence} means that these summaries when applied at finite values of $n$ need not reflect the asymptotic properties as described in $\mbox{IF}_{\alpha,\infty}(\theta)$. We explore this discrepancy below and argue that the asymptotic influence measure is more appropriate in the sense of representing a minimax approach to robustness.  The arguments employed here are broadly applicable to  robust estimators that depend on $n$; in particular, regularized versions of robust estimators are susceptible to the same discrepancies and our analysis provides a framework for describing robustness in this context as well.

\subsection{Breakdown Point}

We begin by motivating a version of the breakdown point for $n$-dependent estimators. Classically, the breakdown point is defined to be
\begin{equation} \label{bpdef}
B(T_n) = \sup \left\{ \alpha: \sup_z | \mbox{IF}_{\alpha,n}(z) | < \infty \right\},
\end{equation}
(see  \citet{Huber81}). \citet{Beran77} and \citet{ParkBasu04} demonstrated that the MDE has a breakdown point of $1/2$ in the case of Hellinger distance and when $G(\cdot)$ and $G'(\cdot)$ are bounded respectively. In contrast, we show in Theorem \ref{zinfinite.thm} and Corollary \ref{breakdown.theorem} below that for each fixed $n$, $B(T_n) = 1$. The distinction between these cases motivates an alternative measure that captures the intuition of the classical breakdown point when applied to a {\em sequence} of estimators that changes over $n$. We define the {\em asymptotic breakdown point} of the sequence $\{T_n\}_{n=1}^{\infty}$, as follows
\begin{equation} \label{abpoint}
B^*\left(\{T_k\}\right) = \sup \left\{ \alpha: \limsup_n \sup_{z} |\mbox{IF}_{\alpha,n}(z)| < \infty \right\}.
\end{equation}
That is, for each $n$ we consider the maximal displacement under $\alpha$-level contamination and declare a breakdown if the limit of these displacements
is unbounded.

It is easy to see that MDAP estimators have asymptotic breakdown 1/2 if their MDE counterparts do and $\log \pi(\ta)$ convex. Writing the MDAP estimator as
\begin{equation} \label{mdap}
\tilde{T}_n(h) = \argmax_{\ta \in \Theta} \left( n D(h,f_\ta) - \log \pi(\ta) \right)
\end{equation}
it is readily seen that if $\ta^*$ maximizes $\pi(\ta)$ then $|\tilde{T}_n(h) - \ta^*| < |\hat{\ta}(h) - \ta^*|$ for $\hat{\ta}(h)$ sufficiently large and thus if $\mbox{IF}_{\alpha,\infty}(z)$ is uniformly bounded, so is the influence function of $\tilde{T}_n$. Nonetheless, the convergence of $\tilde{T}_n(h)$ to $\hat{\ta}(h)$ means that if $\hat{\ta}(h_{\alpha,z_k}) \rightarrow \infty$ there is a sequence $n_k$ so that $T_{n_k}(h_{\alpha,z_{k}}) \rightarrow \infty$ as $k \rightarrow \infty$.

For EDAP estimators a uniform identifiability condition is required. We also impose boundedness on $G$ and $G'$, but note these conditions do not hold for Hellinger distance, however, a direct argument can be given and hence in Theorem \ref{ass.breakdown.thm} below and in other results we will sometimes state that ``The results also hold for Hellinger Distance''.

\begin{thm} \label{ass.breakdown.thm}
Under the conditions of Theorem \ref{EDAP-MDE.thm}, additionally assume $G(\cdot)$ and $G'(\cdot)$ are bounded and let $\int \| \ta \|_2 \pi(\ta) d\ta < \infty$,
then under conditions (\ref{tgortho}-\ref{fgortho}) if,
\begin{equation} \label{BreakdownCondition}
\inf_z \inf_{\ta \in \Theta} D(\alpha t_z,f_\ta) >  \inf_{\ta \in \Theta} D((1-\alpha)g,f_\ta) + \delta,
\end{equation}
for $\alpha \leq 1/2$, the asymptotic breakdown point of $T_1,T_2,\ldots$ is equal to the breakdown point of $\hat{\ta}(h)$. The result continues to hold when $D$ is given by Hellinger distance.
\end{thm}
The identifiability condition imposed above ensures that $t_z$ does not more closely resemble the family $f_\ta$ than $g$. In the analysis of \citet{ParkBasu04}, $g = f_\ta$ is assumed; under these conditions $B^*(\{T_n\}) = 1/2$. If we replace $g$ by the estimate $g_n$ it could happen that  the inequality in \eqref{BreakdownCondition} is reversed \ghadd{(ie for every $z$ there is a $f_{\ta_z}$ closer to $t_z$ than any member of $f_\ta$ is to $g_n$)}. In this case, the breakdown point of $\hat{\ta}(h)$ could be strictly less than 1/2.

These results are in contrast to the treatment of $T_n$ for fixed $n$. Here we follow \citet{Beran77} in evaluating $\lim_{z \rightarrow \infty} T_n(h_{\alpha,z})$ for each $z$.
\begin{thm} \label{zinfinite.thm}
Under the contions of Theorem \ref{ass.breakdown.thm},
\[
\lim_{z \rightarrow \infty} T_n(h_{\alpha,z}) =  T_n( (1-\alpha) g ).
\]
This result also holds for Hellinger distance.
\end{thm}
The condition that $D(g,f_\ta)$ be bounded holds if $|G(\cdot)|$ is bounded; this is assumed in \citet{ParkBasu04} and holds for the negative exponential disparity and Hellinger distance $(0 \leq 2HD(g,f_\ta)\leq 4)$.

For MDE's, taking $g = f_{\ta_0}$ yields $T_n( (1-\alpha) f_{\ta_0}) = \ta_0$. For EDAP estimators, the $(1-\alpha)$ factor generally results in a reduction in strength in the disparity relative to the prior. For Hellinger distance
\[
n 2HD( (1-\alpha) g,f_\ta) = 4 - 4 \sqrt{1-\alpha} \int \sqrt{g(x) f_{\ta}(x)}dx = n \sqrt{1-\alpha} 2HD(g,f) - 4(1-\sqrt{1-\alpha})
\]
since the second term in canceled in normalizing the D-posterior this is equivalent to  reducing $n$ by a factor $\sqrt{1-\alpha}$.

We note here that while the above discussion examines the behavior of $T_n(h_{\alpha,z})$ for small $\alpha$, it can readily be extended to the following corollary
\begin{cor} \label{breakdown.theorem}
Let $D(g,f_\ta)$ be bounded  for all $\ta$ and all densities $g$ and let $\int \| \ta \|_2 \pi(\ta) d\ta < \infty$, then the breakdown point of the EDAP is 1.
\end{cor}
A simple direct proof is given for this in Online Appendix \ref{robust.proofs}. We observe that $D(0,f_\ta) = G(-1)$ is independent of $\ta$, yielding $T_n(0) = \int \ta \pi(\ta) d\ta$: the prior mean.

The results at fixed $n$ indicate an extreme form of robustness that results from the fact that the disparity approximation to the likelihood is weak in its tails. This produces a lack of equivariance in the resulting estimator that appears in the third term of the asymptotic expansion as shown in Equation \ref{expansion.3rdterm} of the online appendix. The fixed $n$ result does not distinguish our estimator from alternative estimators that are clearly problematic.  In particular, the threshold estimator of the mean defined by
\[
m_n(f) = \left\{ \begin{array}{ll} \int x dF(x) & \left| \int x dF(x) \right| < 1000 n \\
1000n * \mbox{sign}\left( \int x dF(x) \right) & \mbox{ otherwise } \end{array} \right.
\]
is also efficient and has breakdown point 1. However, by considering contamination with $t_{z_n}(x)$ taken to be uniform on $[1000n-1, \ 1000n]$ it is readily seen that the asymptotic breakdown point is $B^*(\{m_n\}) = 0$.  We might also contemplate a mean estimate based on a penalized Huber loss:
\[
h_n(f) = \argmin_{\mu} n \int H(x-\mu) dF(x) + \lambda_n \mu^2
\]
where $H$ is the Huber loss function (see \citet{Huber81}) and $\lambda_n \rightarrow 0$. Since $H(z)$ increases linearly for $z$ sufficiently large, the breakdown point of $h(f)$ is also 1, but $B^*(\{h_n\}) = 1/2$. Of course in this case, $h_n(f)$ will not be efficient.

While we have suggested that the distinction between the finite $n$ and asymptotic breakdown point is  a reflection more on the definition \eqref{bpdef} than the properties of EDAP, it does leave considerable room for the design of unbounded disparities that are nonetheless robust and which would therefore also allow the use of improper priors while still obtaining proper D-{\em a posteriori} distributions.

\subsection{Influence Function}

An alternative measure of robustness is given by the influence function \citep{Hampel74}:
\begin{equation} \label{ifdef}
\mbox{IF}_{0,n}(z) = \lim_{\alpha \rightarrow 0} \mbox{IF}_{\alpha,n}(z)
\end{equation}
That this function need not always provide a useful guide to the behavior of $T_n(h_{\alpha,z})$ was observed in \citet{Beran77} and further expanded in \citet{Lindsay94} who demonstrated that all MDE's that yield efficiency share the same influence function as the MLE whatever their behavior at gross levels of contamination. The analysis in \citet{Lindsay94} implicitly assumes an equivariant estimator so that $T_n(f_{\theta}) = \theta$ for any $\theta$. When the effect of a prior is included in the analysis a different result is obtained at finite samples, but an equivalent limiting result can be derived.

To examine the influence function for EDAP estimators, we assume that the limit may be taken inside
all integrals in (\ref{ifdef}) and obtain
\begin{align*}
\mbox{IF}_{0,n}(z) & = n E_{P_D(\ta|g)}\left[  \ta  C_{z}(\ta,g) \right] -
 n \left[ E_{P_D(\ta|g)} \ta \right] \left[ E_{P_D(\ta|g)} C_{z}(\ta,g) \right] \\
 & = n \mbox{Cov}_{P_D(\ta|g)}\left(\ta, C_{z}(\ta,g)\right).
\end{align*}
where $E_{P_D(\ta|g)}$ indicates expectation with respect to the D-posterior with
fixed density $g$ and
\begin{align*}
C_{z}(\ta,g) & = \left. \frac{d}{d\epsilon}  \int G\left(
\frac{h_{z,\epsilon}(x)}{f_\ta(x)}-1 \right) f_\ta(x) dx \right|_{\epsilon = 0} \\
& =  \int G'\left(\frac{g(x)}{f_\ta(x)}-1 \right)(t_z(x) - g(x)) dx.
\end{align*}
Here we observe that $\mbox{IF}_{0,n}(z)$ depends on the prior $\pi$. This is the case for any {\em a posteriori} estimate. $\mbox{IF}_{0,n}(z)$ also depends on the disparity employed as demonstrated in
\begin{thm} \label{influence.theorem}
Let $D(g,f_\ta)$ be bounded and assume that
\begin{equation} \label{influence1}
e_0 = \sup_x \int \left| G'\left(\frac{g(x)}{f_\ta(x)}-1\right) \pi(\ta)
\right| d\ta  < \infty \mbox{ and } e_1 = \sup_x \int \left| \ta
G'\left(\frac{g(x)}{f_\ta(x)}-1\right) \pi(\ta) \right| d\ta  < \infty
\end{equation}
then $|IF(\ta;g,t_z)| < \infty$.
\end{thm}
In the case of Hellinger distance the conditions of Theorem
\ref{influence.theorem} require the boundedness of $r(x) = \int
(\sqrt{f_\ta(x)}/\sqrt{g(x)}) \pi(\theta) d\theta$ which may not always hold \citep[e.g.][]{Beran77}.

Despite this strong result, in an asymptotic sense the choice of disparity and prior is indestiguishable when $g$ is assumed to lie within the model class. Expanding $C_{z}(\ta,g)$ about $\bar{\ta} = E_{P_D(\ta|g)}\ta$ provides
\begin{align*}
\mbox{IF}_{0,n}(z) & = n E_{P_D(\ta|g)}\left(\ta - \bar{\ta}\right)^2 C_{z}'(\bar{\ta}) + \frac{n}{2} E_{P_D(\ta|g)} \left(\ta - \bar{\ta}\right)^3 C_{z}''(\ta^*) \\
& = I_{D}(\theta_g)^{-1} C_z'(\ta_g) + (\bar{\ta}-\ta_g)C_z'(\ta^+) +  o(n^{-1/2}) \\
& =   I_{D}(\theta_g)^{-1} C_z'(\ta_g) +  o(n^{-1/2})
\end{align*}
where $\ta^*$ is lies between $\ta$ and $\bar{\ta}$ and $\ta^+$ between $\bar{\ta}$ and $\ta_g$, since $(\ta - \ta_g)$ and $(\bar{\ta}-\ta_g)$ are $O_p(n^{-1/2})$ and $o(n^{-1/2})$ respectively. We now observe that when $g = f_{\ta_0}$ is in the model class,
\[
 C_z'(\ta_0) = \int \frac{\nabla_\ta f_{\ta_0}(x)}{f_{\ta_0}(x)}(t_z(x)-g(x)) dx
\]
is independent of both the disparity and the prior, as is $I_D(\theta_0)$ and the limiting value of $\mbox{IF}_{0,n}(z)$ coincides with the influence function of the MLE.  We also note that the next leading term in the expansion above is $C_z''(\ta)$, which co-incides with the second-order approximation in \citet[][Eqn. 7]{Lindsay94}. Unlike the case of the MDE, however, here the second order term does affect the influence function at finite $n$.

\section{Efficiency and Numerical Results} \label{sec:simple}

While there is a large literature on robust estimation methods, disparity-based estimation methods also achieve statistical
efficiency when $g$ is a member of the parametric family $f_\ta$.  In this section, we present theoretical results for i.i.d. data to demonstrate that
inference based on the D-posterior is also asympotically efficient. We also conduct a simulation study to demonstrate the finite-sample
performance of these estimators.

\subsection{Efficiency}

We recall that under suitable regularity conditions, expected {\em a
posteriori} estimators are strongly consistent, asymptotically normal and are
statistically efficient; \citep[see][Theorems 4.2-4.3]{GDS06}. Our results in
this section show that this property continues to hold for EDAP estimators
under regularity conditions on $G(\cdot)$ when the model $\{f_\ta: \ta \in
\Theta\}$ contains the true distribution.  We define
\[
I^D(\ta) = \nabla^2_\ta D(g,f_{\ta}), \mbox{ and } \hat{I}_n^D(\ta)
= \nabla^2_\ta D(g_n,f_{\ta})
\]
as the disparity information and $\ta_g$ the parameter that minimizes $D(g,f_\ta)$ (note
that $\ta_g$ here depends on $g$).
We note that if $g = f_{\ta_g}$, $I^D(\ta_g)$ is
exactly equal to the Fisher information for $\ta_g$.

The proofs of our asymptotic results rely on the assumptions listed below.
Among these are that minimum disparity estimators are strongly consistent and
efficient; this in turn relies on further assumptions, some of which make those
listed below redundant. They are given here to maximize the mathematical
clarity of our arguments.  We assume that $X_1,\ldots,X_n$ are i.i.d. generated
from some distribution $g(x)$ and that a parametric family, $f_{\ta}(x)$ has
been proposed for $g(x)$ where $\ta$ has distribution $\pi$.   To demonstrate
efficiency, we assume
\begin{enumerate}[({A}1)]
\item $g(x) = f_{\ta_g}(x)$; i.e. $g$ is a member of the parametric family.
\label{G.param}

\item $G$ has three continuous derivatives with $G'(0) = 0$, $G''(0) = 1$ and
$|G'''(0)| \leq \infty$. \label{G.ass}

\item \ghadd{ There exists $C > 0$ such that for all $g$ and $h$
\[
\sup_{\theta \in \Theta} | D(g,f_\ta) - D(h,f_\ta) | \leq C \int |g(x)-h(x)| dx.
\]
} \label{l1continuity}

\item  $\nabla^2_\ta D(g,f_{\ta})$ is positive definite and continuous in
$\ta$ at $\ta_g$ and continuous in $g$ with respect to the $L_1$ metric. \label{ass.smooth}

\item For any $\delta > 0$, there exists $\epsilon>0$ such that
\[
\sup_{|\ta - \ta_g|> \delta}( D(g,f_{\ta})-D(g,f_{\ta_g})) > \epsilon
\] \label{identifiability}


\item  The minimum disparity estimator, $\hat{\ta}_n$, satisfies
$\hat{\ta}_n \rightarrow \ta_g$ almost surely and $\sqrt{n}(\hat{\ta}_n -
\ta_g) \convd N(0,I^D(\ta)^{-1})$. \label{MDE.efficiency}
\end{enumerate}

Our first result concerns the limit distribution for the posterior density of
$\sqrt{n}(\ta-\hat{\ta}_n)$, which demonstrates that the D-posterior
converges in $L_1$ to a Gaussian density centered on the minimum disparity
estimator $\hat{\ta}_n$ with variance $\left[nI^D(\hat{\ta}_n)\right]^{-1}$. This establishes that credible intervals based on
either $P_D(\ta|x_1,\ldots,x_n)$ or from
$N(\hat{\ta}_n,I^D_n(\hat{\ta}_n)^{-1})$ will be asymptotically accurate.

\begin{thm} \label{l1convergence}
Let $\hat{\ta}_n$ be the minimum disparity estimator of $\ta_g$, $\pi(\ta)$ be
any prior that is continuous and positive at $\ta_g$ with $\int_\Theta
\|\ta\|_2 \pi(\ta) d\ta < \infty$ where $\|\cdot\|_2$ is the usual 2-norm, and
$\pi^{*D}_n(t)$ be the D-posterior density of $t = (t_1,\ldots,t_p) =
\sqrt{n}(\ta-\hat{\ta}_n)$. Then, under conditions
(A\ref{G.ass})-(A\ref{MDE.efficiency}),
\begin{equation} \label{eq:l1convergence}
\lim_{n \rightarrow \infty} \int \left| \pi^{*D}_n(t) -
\left(\frac{|I^D(\ta_g)|}{2 \pi}\right)^{p/2} e^{-\frac{1}{2}t' I^D(\ta_g) t} \right|
dt \convas 0.
\end{equation}
Furthermore, \eqref{eq:l1convergence} also holds with $I^D(\ta_g)$
replaced with $\hat{I}^D_n(\hat{\ta}_n)$.
\end{thm}

Our next theorem is concerned with  the efficiency
and asymptotic normality of EDAP estimates.

\begin{thm} \label{posteriorefficiency}
Assume conditions (A\ref{G.ass})-(A\ref{MDE.efficiency}) and  $\int_\Theta
\|\ta\|_2 \pi(\ta) d\ta < \infty$, then \\ $\sqrt{n} \left( \ta^*_n - \hat{\ta}_n
\right) \convas 0$ where $\ta^*_n$ is the EDAP estimate. Further,
$\sqrt{n}\left(\ta^*_n - \ta_g \right) \convd N\left(0,I^D(\ta_g)\right)$.
\end{thm}

The proofs of these theorems are deferred to the online appendix
\ref{efficiency.proofs}, but the following remarks concerning the assumptions (A1)-(A6) are in order:
\begin{enumerate}
\item Assumption A\ref{G.param} states that $g$ is a member of the parametric
family. When this does not hold, a central limit theorem can be derived for
$\hat{\ta}_n$ but the variance takes a sandwich-type form; see \citet{Beran77}
in the case of Hellinger distance. For  brevity, we have followed \citet{BSV97}
and \citet{ParkBasu04} in restricting to the parametric case.

\item Assumptions A\ref{G.ass}-A\ref{identifiability} are required for the regularity and
identifiability of the parametric family $f_{\ta}$ in the disparity $D$. \ghadd{Note that A\ref{l1continuity} holds for Hellinger distance and if $G'(\cdot)$ is bounded from arguments in \citet{ParkBasu04}; other disparities may require specialized demonstrations.}
Specific conditions for A\ref{MDE.efficiency} to hold are given in various
forms in \citet{Beran77,BSV97,ParkBasu04} and \citet{ChenVidyashankar06}, see
conditions in Online Appendix \ref{efficiency.proofs}.


\item The proofs of these results
employ the same strategies as those for posterior asymptotic efficiency
\citep[see][for example]{GDS06}. However, here we rely on the second-order convergence
of the disparity to the likelihood at appropriate rates and the consequent asymptotic
efficiency of minimum-disparity estimators, which in turn is based on a careful analysis of non-parametric density estimates.

\item Since the structure of the proof only requires  second-order
properties and appropriate rates of convergence, we can replace $D(g_n,f_{\ta})$ for
i.i.d. data with an appropriate disparity-based term for more complex
models as long as A\ref{MDE.efficiency}
can be shown hold.  In particular, the
results in \citet{HookerVidyashankar09}
suggest
that the disparity methods for regression problems detailed in Section
\ref{sec:conditional}
will also yield efficient estimates.
\end{enumerate}

\subsection{Simulation Studies}

To illustrate the small sample performance of D-posteriors, we undertook a
simulation study for i.i.d. data from Gaussian distribution. 1,000 sample data sets of size 20 from
a $N(5,1)$ population were generated. For each sample data set, a random walk
Metropolis algorithm was run for 20,000 steps using a $N(0,0.5)$ proposal
distribution and a $N(0,25)$ prior, placing the true mean one prior standard deviation above the prior mean. The kernel bandwidth was selected by the
bandwidth selection in \citet{SheatherJones91}. H- and N-posteriors were easily
calculated by combining the \texttt{KernSmooth} \citep{KernSmooth} and
\texttt{LearnBayes} \citep{LearnBayes} packages in \texttt{R}. We also report an experiment
in which the normal log likelihood is replaced in the posterior with Tukey's biweight objective function using a cut-point of
4.685 as a  comparison to alternative robust estimators. \ghadd{In order to compare computational cost, we have run an MCMC chain for the normal log likelihood and report these below, even though analytic posteriors are available.}

Expected {\em a
posteriori} estimates for the sample mean were obtained along with 95\%
credible intervals from every second sample in the second half of the MCMC chain.
Outlier contamination was investigated by reducing the last one, two or five
elements in the data set by 3, 5 or 10. This choice was made so that both outliers and
prior influence the EDAP in the same direction. The analytic
posterior without the outliers is normal with mean 4.99 (equivalently, bias
of -0.01) and standard deviation 0.223.

The results of this simulation are
summarized in Tables \ref{simtable1.paper} (uncontaminated data) and \ref{simtable2.paper} (contaminated data).
As can be expected, the standard Bayesian
posterior suffers from sensitivity to large negative values whereas the disparity-based
methods remain nearly unchanged. Tukey's biweight also ignored large outliers, but was more sensitive
than the disparity methods to larger amounts of contamination.  Near-outliers at the smaller value of -3
resulted in similar biases across all methods.   We observe that all robust estimates have slightly larger standard
deviations than the EAP corresponding to a loss of efficiency of 2\% for the H-posterior and 5\% for the N-posterior and Tukey estimates.
We speculate the increased variance from the N-posterior is due to its relatively Heavier tails (the maximal value of NED is $e^{-1}$ compared
to 4 for 2HD).  A comparison of CPU time
indicates that the use of disparity methods required a little more than twice
the computational effort as compared to using the likelihood within an MCMC
method. Further details from this simulation, including comparisons with Huber estimators are given in Online Appendix \ref{sim:iid}.

The influence of the prior is investigated in the right-hand plot of Figure
\ref{likapprox} where we observe that the EAP and EDAP estimates are
essentially identical until the prior is about 9 standard deviations from the
mean of the data: at this point the prior dominates. However, we note that this
picture will depend strongly on the prior chosen; a less informative prior will
have a smaller range of dominance.

Because the normal distribution is symmetric, estimating its mean is relatively easy. We therefore also conducted a
simulation to estimate both shape and scale parameters in an \ghadd{exponential}-Gamma distribution \ghadd{(i.e. $\exp(X_i)$ has a Gamma distribution)}. The details and results of this simulation are
reserved to Online Appendix \ref{sim:iid}. We observed the expected behavior: EDAP estimates remained insensitive to outliers, whereas
they significantly distorted the EAP. However in this case, the H-posterior demonstrated larger variance than the N-posterior which we
explain as being due to the tendency of nonparametric density estimates from exponential-Gamma data to become bimodal: producing inliers where a large value of the parametric density is compared to a relatively small value of the nonparametric estimate.

\begin{table}
\caption{\label{simtable1.paper} A simulation study for a normal mean using the usual posterior, the
Hellinger posterior and the Negative Exponential posterior. Columns  give the bias and variance of the posterior mean, coverage and
average CPU time of the central 95\% credible interval based on 1,000
simulations. These are recorded for the posterior, Hellinger
distance (HD), negative exponential disparity (NED) and Tukey's biweight objective used in place
of the log likelihood (Biweight).}
\centering
\fbox{
\begin{tabular}{l|ccccc}
                     &   Bias   & SD    & Coverage & Length & CPU Time \\ \hline
Posterior            & -0.015   & 0.222 &   0.956  & 0.873  & 3.393 \\
Hellinger            & -0.015   & 0.225 &   0.954  & 0.920  & 7.669 \\
Negative Exponential & -0.018   & 0.229 &   0.973  & 1.022  & 7.731 \\
Biweight     & -0.017 & 0.228 &   0.977  & 1.007  & 3.523 \\ \hline
\end{tabular}}
\end{table}

\begin{table}
\caption{\label{simtable2.paper} Results for contaminating the data sets used in Table \ref{simtable1.paper} with outliers.
1, 2, and 5 outliers (large
columns) are added at locations -3, -5 and -10 (column Loc) for the posterior, Hellinger
distance (HD), negative exponential disparity (NED) and Tukey's biweight objective used in place
of the log likelihood (Biweight).}
\centering
\fbox{
\begin{tabular}{l||ccc|ccc|ccc}
&  \multicolumn{3}{c}{1 Outlier} &  \multicolumn{3}{c}{2 Outliers} &
\multicolumn{3}{c}{5 Outliers} \\ \hline
Loc  & Bias    & SD    & Coverage & Bias   & SD    & Coverage & Bias   & SD    & Coverage \\ \hline \hline
Posterior & & & & & & & & & \\
-3     & -0.164  & 0.219  & 0.883   & -0.300 & 0.206 & 0.722    & -0.637 & 0.182 & 0.100 \\
-5     & -0.264  & 0.219  & 0.778   & -0.490 & 0.206 & 0.375    & -1.053 & 0.182 & 0.001 \\
-10    & -0.513  & 0.219  & 0.360   & -0.965 & 0.207 & 0.004    & -2.093 & 0.182 & 0.000 \\ \hline
HD & & & & & & & & & \\
-3     &  -0.109 & 0.246 & 0.920    & -0.194 & 0.275 & 0.859    & -0.237 & 0.299  & 0.770 \\
-5     &  -0.027 & 0.238 & 0.942    & -0.040 & 0.257 & 0.928    & -0.024 & 0.305  & 0.865 \\
-10    &  -0.014 & 0.234 & 0.948    & -0.019 & 0.249 & 0.935    &  0.018 & 0.286  & 0.883 \\ \hline
NED & & & & & & & & & \\
-3     &  -0.080 & 0.256 & 0.959    & -0.133 & 0.279 & 0.933    & -0.166 & 0.308  & 0.893 \\
-5     &  -0.020 & 0.238 & 0.977    & -0.025 & 0.243 & 0.968    & -0.015 & 0.264  & 0.948 \\
-10    &  -0.017 & 0.237 & 0.973    & -0.020 & 0.241 & 0.970    & -0.007 & 0.260  & 0.952 \\ \hline
Biweight & & & & & & & & & \\
-3     & -0.091 & 0.246 & 0.954    & -0.175 & 0.252 & 0.915     & -0.443 & 0.275 & 0.645 \\
-5     & -0.018 & 0.237 & 0.974    & -0.019 & 0.236 & 0.972     & -0.022 & 0.243 & 0.967 \\
-10    & -0.017 & 0.236 & 0.977    & -0.018 & 0.234 & 0.971     & -0.018 & 0.236 & 0.969 \\
\end{tabular} }
\end{table}

\section{Disparities based on Conditional Density for Regression Models}
\label{sec:conditional}

The discussion above, along with most of the literature on disparity
estimation, has focussed on i.i.d. data in which a kernel density estimate may
be calculated. The restriction to i.i.d. contexts severely limits the
applicability of disparity-based methods. We extend these methods to non-i.i.d.
data settings via the use of conditional density estimates. This extension is
studied in the frequentist context in the case of minimum-disparity estimates
for parameters in non-linear regression in \citet{HookerVidyashankar09b}.

Consider the classical regression framework with data $(Y_1,X_1),\ldots,(Y_n,X_n)$ is a
collection of i.i.d. random variables where
inference is made conditionally on $X_i$.  For continuous $X_i$, a
non-parametric estimate of the conditional density of $y|x$ is given by \citet{Hansen04} and \citet{LiRacine07}:
\begin{equation} \label{conditional.kernel}
g_n^{(c)}(y|x) = \frac{ \frac{1}{n c_{n1} c_{n2}} \sum_{i=1}^n K\left( \frac{y -
Y_i}{c_{n1}} \right) K \left( \frac{\|x-X_i\|}{c_{n2}} \right) }{ \frac{1}{n
c_{n2}} \sum_{i=1}^n K \left( \frac{\|x-X_i\|}{c_{n2}} \right) }.
\end{equation}
Under a parametric model $f_\ta(y|X_i)$ assumed  for the conditional distribution of $Y_i$
given $X_i$, we define a disparity between $g_n^{(c)}$ and $f_\ta$ as follows:
\begin{equation} \label{cdisp}
D^{(c)}(g_n^{(c)},f_{\ta}) = \sum_{i=1}^n D\left( g_n^{(c)}(\cdot|X_i),f_{\ta}(\cdot|X_i)\right).
\end{equation}
As before, for Bayesian inference we replace the log likelihood by negative of the conditional
disparity \eqref{cdisp}; that is,
\[
e^{l(Y_|X_i,\ta)} \pi(\ta) \approx e^{-D^{(c)}(g_n^{(c)},f_{\ta})} \pi(\ta).
\]
In the case of simple linear regression, $Y_i = \beta_0 + \beta_1 X_i +
\epsilon_i$, $\ta = (\beta_0,\beta_1,\sigma^2)$ and  $f_{\ta}(\cdot|X_i) =
\phi_{\beta_0+\beta_1 X_i,\sigma^2} (\cdot)$ where $\phi_{\mu,\sigma^2}$ is
Gaussian density with mean $\mu$ and variance $\sigma^2$.

The use of a conditional formulation, involving a density estimate over a
multidimensional space, produces an asymptotic bias in MDAP and EDAP estimates similar to that found in
\citet{TamuraBoos86}, who also note that this bias is generally small. Online Appendix
\ref{sec:reduction} proposes two alternative formulations that reduce the
dimension of the density estimate and the bias.

When the $X_i$ are discrete, (\ref{conditional.kernel})
reduces to a distinct conditional density for each level of $X_i$. For
example, in a one-way ANOVA model $Y_{ij} = X_i + \epsilon_{ij}$, $j =
1,\ldots,n_i$, $i = 1,\ldots,N$, this reduces to
\[
g_n^{(c)}(y|X_i) = \frac{1}{n_i c_n} \sum_{i=1}^{n_i} K \left( \frac{y - Y_{ij}}{c_n} \right).
\]
We note that in this case the bias noted above does not appear. However When
the $n_i$ are small, or for high-dimensional covariate spaces the
non-parametric estimate $g_n(y|X_i)$ can become inaccurate. The marginal
methods discussed in Online Appendix \ref{sec:reduction} can also be employed in this
case.

Online Appendix \ref{sim:linreg} gives details of a simulation study of this method as well
as those described in Online Appendix \ref{sec:reduction} for a regression
problem with a three-dimensional covariate.  All disparity-based methods
perform similarly to using the posterior with the exception of the conditional
form in Section \ref{sec:conditional} when Hellinger distance is used which
demonstrates a substantial increase in variance. We speculate that this is due
to the sparsity of the data in high dimensions creating inliers; negative
exponential disparity is less sensitive to this problem \citep{BSV97}.

\section{Disparity Metrics and the Plug-In Procedure} \label{sec:plugin}

The disparity-based techniques developed above can be extended to hierarchical models.
In particular, consider the following structure
for an observed data vector $Y$ along with an unobserved latent effect vector
$Z$ of length $n$:
\begin{equation} \label{hierarchicalmodel}
P(Y,Z,\ta) = P_1(Y|Z,\ta)P_2(Z|\ta)P_3(\ta)
\end{equation}
where $P_1$, $P_2$ and $P_3$ are the conditional distributions of $Y$ given $Z$
and $\theta$ the distribution of $Z$ given $\theta$ and the prior distribution
of $\theta$.  Any term in this factorization that can be expressed as the
product of densities of i.i.d. random variables can now be replaced by a
suitably chosen disparity. This creates a {\em plug-in procedure} in which
particular terms of a complete data log likelihood are replaced by disparities.
For example, if the middle term is assumed to be a product:
\[ P(Z|\ta) = \prod_{i=1}^n p(Z_i|\ta), \]
inference can be robustified for the distribution of the $Z_i$ by replacing
(\ref{hierarchicalmodel}) with
\[ P_{D_1}(Y,Z,\ta) = P(Y|Z,\ta)e^{-2D(g_{n}(\cdot;Z),P_2(\cdot|\ta))} P_3(\ta)
\]
where
\[
g_{n}(z;Z) = \frac{1}{nc_n} \sum_{i=1}^n K\left(\frac{z-Z_i}{c_n}\right).
\]
In an MCMC scheme, the $Z_i$ will be imputed at each iteration and the estimate
$g_n(\cdot;Z)$ will change accordingly. If the integral is evaluated using
Monte Carlo samples from $g_n$, these will also need to be updated. The
evaluation of $D(g_{n}(\cdot;Z),P_2(\cdot|\ta))$ creates additional
computational overhead, but we have found this to remain feasible for moderate
$n$. A similar substitution may also be made for the first term using the
conditional approach suggested above.

To illustrate this principle in a concrete example, consider  a one-way random-effects model:
\[ Y_{ij} = Z_i + \epsilon_{ij}, \ i = 1,\ldots,n, \ j = 1,\ldots,n_i \]
under the assumptions
\[ \epsilon_{ij} \sim N(0,\sigma^2), \ Z_i \sim N(\mu,\tau^2) \]
where the interest is in the value of $\mu$. Let $\pi(\mu,\sigma^2,\tau^2)$ be
the prior for the parameters in the model; an MCMC scheme may be conducted with
respect to the probability distribution
\begin{equation} \label{rand.effects.density}
P(Y,Z,\mu,\sigma^2,\tau^2) = \prod_{i=1}^n \left(\prod_{j=1}^{n_i}
\phi_{0,\sigma^2} (Y_{ij}-Z_i) \right) \prod_{i=1}^n \phi_{\mu,\tau^2}(Z_i)
\pi(\mu,\sigma^2,\tau^2)
\end{equation}
where $\phi_{\mu,\sigma^2}$ is the $N(\mu,\sigma^2)$ density. There are now two
potential sources of distributional errors: either in individual observed
$Y_{ij}$, or in the unobserved $Z_i$.  Either (or both) possibilities can be
dealt with via the plug-in procedure described above.

If there are concerns that the distributional assumptions on the $\epsilon_{ij}$ are
not correct, we observe that the statistics $Y_{ij}-Z_i$ are assumed to be
i.i.d. $N(0,\sigma^2)$. We may then form the conditional kernel density
estimate:
\[
g_{n}^{(c)}(t|Z_i;Z) = \frac{1}{n c_{n1}} \sum_{j=1}^{n_i} K\left(
\frac{t - (Y_{ij}-Z_i)}{c_{n1}} \right)
\]
and replace (\ref{rand.effects.density}) with
\begin{equation} \label{rand.effects.disparity1}
P_{D_2}(Y,Z,\mu,\sigma^2,\tau^2) = e^{-\sum_{i=1}^n n_i
D(g_{n}^{(c)}(t|Z_i;Z),\phi_{0,\sigma^2}(\cdot))} \prod_{i=1}^n
\phi_{\mu,\tau^2}(Z_i)  \pi(\mu,\sigma^2,\tau^2).
\end{equation}

On the other hand, if the distribution of the $Z_i$ is miss-specified, we form the estimate
\[
g_{n}(z;Z) = \frac{1}{n c_{n2}} \sum_{i=1}^n K\left( \frac{z-Z_i}{c_{n2}} \right)
\]
and use
\begin{equation} \label{rand.effects.disparity2}
P_{D_1}(X,Y,\mu,\sigma^2,\tau^2) = \prod_{i=1}^n \left(\prod_{j=1}^{n_i}
\phi_{0,\sigma^2} (Y_i-Z_i) \right) e^{-n
D(g_{n}(\cdot;Z),\phi_{\mu,\tau^2}(\cdot))} \pi(\mu,\sigma^2,\tau^2)
\end{equation}
as the D-posterior.
For inference using this posterior, both $\mu$ and the $Z_i$ will be included as
parameters in every iteration, necessitating the update of $g_n(\cdot;Z)$ or
$g_n^{(c)}(\cdot|z;Z)$. Naturally, it is also possible to
substitute a disparity in both places:
\begin{equation} \label{rand.effects.disparity3}
P_{D_{12}}(Z,Y,\mu,\sigma^2,\tau^2) = e^{- \sum_{i=1}^n n_i   D(g_{n}^{(c)}(\cdot|Z_i;Z),\phi_{0,\sigma^2}(\cdot))}
 e^{-n  D(g_{n}(\cdot;Z),\phi_{\mu,\tau^2}(\cdot))} \pi(\mu,\sigma^2,\tau^2).
\end{equation}

A simulation study considering all these approaches with Hellinger distance
chosen as the disparity is described in  Online Appendix \ref{sim:randeffect}.
Our results indicate that all replacements with disparities perform well,
although some additional bias is observed in the estimation of variance
parameters which we speculate to be due to the interaction of the small sample
size with the kernel bandwidth. Methods that replace the random effect
likelihood with a disparity remain largely unaffected by the addition of an
outlying random effect while for those that do not the estimation of both the
random effect mean and variance is substantially biased.

While a formal analysis of this method is beyond the scope of this paper we
remark that the use of density estimates of latent variables requires
significant theoretical development in both Bayesian and frequentist contexts.
In particular, in the context of using $P_{D_1}$ appropriate inference on
$\theta$ will require local agreement in the integrated likelihoods
\begin{align*}
& \int\cdots\int \prod_{i=1}^n \left(\prod_{j=1}^{n_i} \phi_{0,\sigma^2}
(Y_i-Z_i) \right) e^{-n D(g_{n}(\cdot;Z),\phi_{\mu,\tau^2}(\cdot))}
dZ_1,\ldots,dZ_n \\ & \hspace{2cm} \approx \int\cdots\int \prod_{i=1}^n
\left(\prod_{j=1}^{n_i} \phi_{0,\sigma^2} (Y_{ij}-Z_i) \right) \prod_{i=1}^n
\phi_{\mu,\tau^2}(Z_i) dZ_1,\ldots,dZ_n.
\end{align*}
This can be demonstrated if the $n_i \rightarrow \infty$ and hence the
conditional variance of the $Z_i$ is made to shrink at an appropriate rate.

We note here that the Bayesian methods developed in this paper are particularly relevant in allowing the use of MCMC
for these problems. A frequentist analysis could be obtained by marginalizing over the $Z_i$ in \eqref{rand.effects.disparity1}, \eqref{rand.effects.disparity2}, or \eqref{rand.effects.disparity3}. However, this marginalization is numerically challenging while it can be
very readily obtained in a Bayesian context via MCMC methods.

\section{Real Data Examples} \label{sec:examples}

\subsection{Parasite Data} \label{sec:parasite}

We begin with a one-way random effect model for binomial data.  These data come
from one equine farm participating in a parasite control study in Denmark in
2008. Fecal counts of eggs of the Equine Strongyle parasites were taken pre-
and post- treatment with the drug Pyrantol; the full study is presented in
\citet{NVHPK10}.  The data used in this example are reported in Online Appendix
\ref{sec:data}.

For our purposes, we model the post-treatment data from each horse as binomial
with probabilities drawn from a logit normal distribution. Specifically, we
consider the following model:
\[
k_i \sim \mbox{Bin}(N_i,p_i), \ \mbox{logit}(p_i) \sim N(\mu,\sigma^2), \ i =
1,\ldots,n,
\]
where $N_i$ are the pre-treatment egg counts and $k_i$ are the post-treatment
egg counts.
We observe the data $(k_i,N_i)$ and desire an
estimate of $\mu$ and $\sigma$. The likelihood for these data are
\[
l(\mu,\sigma|k,N) = -\sum_{i=1}^n \left[ k_i \log p_i + (N_i-k_i) \log(1-p_i)
\right] - \frac{1}{2\sigma^2} \sum_{i=1}^n \left( \log(p_i) - \mu \right)^2.
\]
We cannot use conditional disparity methods to account for outlying $k_i$ since
we have only one observation per horse. However, we can consider robustifying
the $p_i$ distribution by use of a negative exponential disparity:
\begin{align*}
g_n(\lambda;p_1,\ldots,p_n) & = \frac{1}{nc_n} \sum K\left( \frac{\lambda -
\mbox{logit}(p_i)}{c_n} \right) \\
l^N(\mu,\sigma|k,N) & = -\sum_{i=1}^n \left[ k_i \log p_i + (N_i-k_i) \log(1-p_i)
\right] - n D(g_n(\cdot;p_1,\ldots,p_n),\phi_{\mu,\sigma^2}(\cdot))
\end{align*}
In order to perform a Bayesian analysis, $\mu$ was given a $N(0,5)$ prior and
$\sigma^2$ an inverse Gamma prior with shape parameter 3 and scale parameter
0.5.  These were chosen as conjugates to the assumed Gaussian distribution and
are defuse enough to be relatively uninformative while providing reasonable
density at the maximum likelihood estimates. A random walk Metropolis algorithm
was run for this scheme with parameterization \\
$(\mu,\log(\sigma),\mbox{logit}(p_1),\ldots,\mbox{logit}(p_n))$ for 200,000
steps with posterior samples collected every 100 steps in the second half of
the chain. $c_n$ was chosen via the method in \citet{SheatherJones91} treating
the empirical probabilities as data.

The resulting posterior distributions, given in Figure \ref{parasite.example},
indicate a substantial difference between the two posteriors, with the
N-posterior having higher mean and smaller variance. This suggests some outlier
contamination and a plot of a sample of densities $g_n$ on the right of Figure
\ref{parasite.example} suggests a lower-outlier with $\mbox{logit}(p_i)$ around
-4. In fact, this corresponds to observation 5 which had unusually high
efficacy in this horse. Removing the outlier results in good agreement between
the posterior and the N-posterior. We note that, as also observed in
\cite{Stigler73}, trimming observations in this manner, unless done carefully,
may not yield accurate credible intervals.

\begin{figure}
\begin{center}
\begin{tabular}{ccc}
\includegraphics[height=4cm]{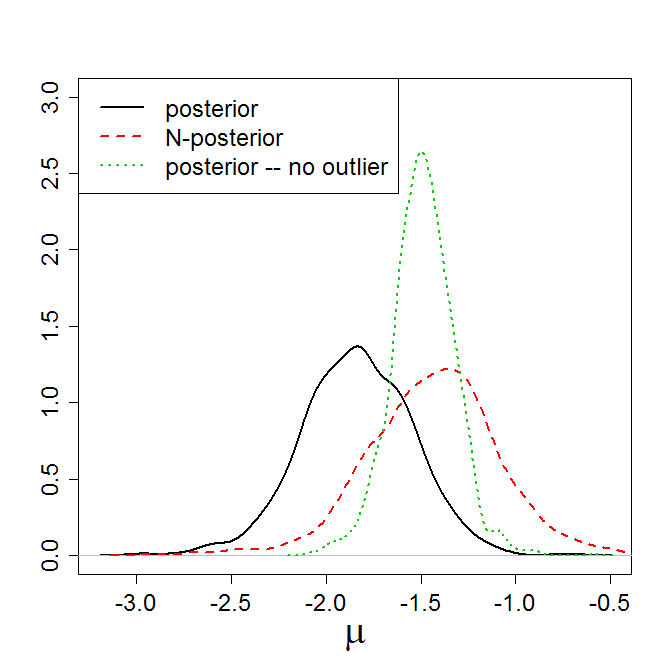} &
\includegraphics[height=4cm]{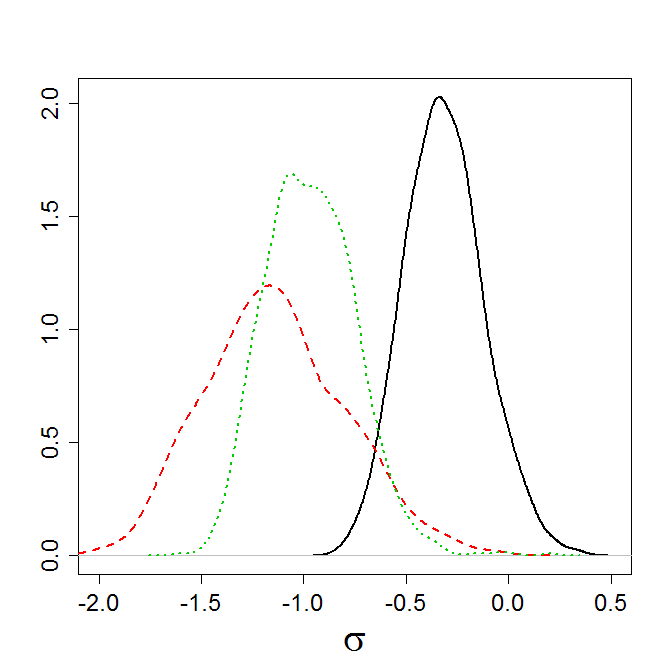} &
\includegraphics[height=4cm]{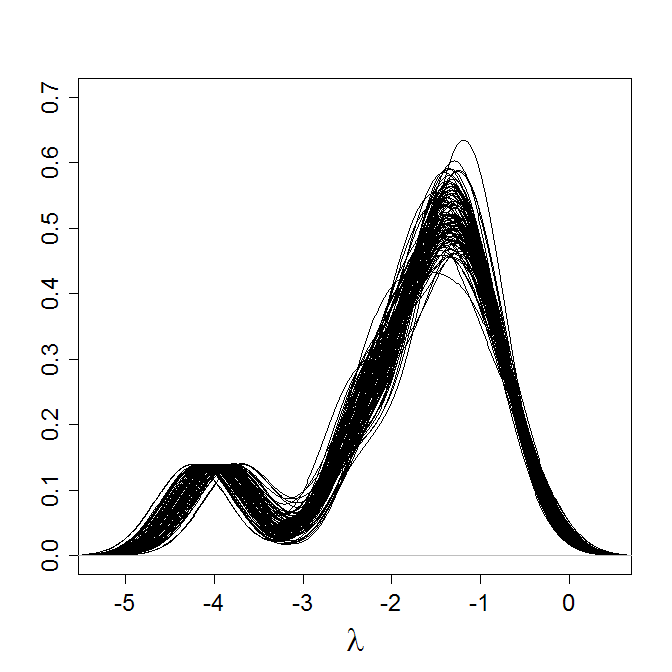}
\end{tabular}
\end{center}
\caption{ Posterior distributions for the parasite data. Left: posteriors for
$\mu$ with and without an outlier and the N-posterior. Middle: posteriors for
$\sigma$. Right: samples of $g_n$ based on draws from the posterior for
$\lambda_1,\ldots,\lambda_n$, demonstrating an outlier at -4.}
\label{parasite.example}
\end{figure}

\subsection{Class Survey Data} \label{sec:survey}

Our second data set are from an in-class survey in an introductory statistics
course held at Cornell University in 2009. Students were asked to specify their
expected income at ages 35, 45, 55 and 65.  Responses from 10 American-born and
10 foreign-born students in the class are used as data in this example; the
data are presented and plotted in  Online Appendix \ref{sec:data}. Our object
is to examine the expected rate of increase in income and any differences in
this rate or in the over-all salary level between American and foreign
students. From the plot of these data in Figure \ref{surveydata.fig} in Online
Appendix \ref{sec:data} some potential outliers in both over-all level of
expected income and in specific deviations from income trend are evident.

This framework leads to a longitudinal data model. We begin with a random
intercept model
\begin{equation} \label{randintercept}
Y_{ijk} = b_{0ij} + b_{1j} t_k + \epsilon_{ijk}
\end{equation}
where $Y_{ijk}$ is log income for the $i$th student in group $j$ (American
($a$) or foreign ($f$)) at age $t_k$. We extend to this the distributional
assumptions
\[
b_{0ij} \sim N(\beta_{0j},\tau_0^2), \ \epsilon_{ijk} \sim N(0,\sigma^2)
\]
leading to a complete data log likelihood given up to a constant by
\begin{equation} \label{randintercept.cdll}
l(Y,\beta,\sigma^2,\tau_0^2) = -\sum_{i=1}^n \sum_{j \in \{a,f\}} \sum_{k=1}^4
\frac{1}{2\sigma^2}\left(Y_{ijk} - b_{0ij} - \beta_{1j} t_k \right)^2 -
\sum_{i=1}^n \sum_{j \in \{a,f\}} \frac{1}{2 \tau_0^2} \left( b_{0ij} -
\beta_{0j} \right)^2
\end{equation}
to which we attach Gaussian priors centered at zero with standard deviations
150 and 0.5 for the $\beta_{0j}$ and $\beta_{1j}$ respectively and Gamma priors
with shape parameter 3 and scale 0.5 and 0.05 for $\tau_0^2$ and $\sigma^2$.
These are chosen to correspond to the approximate orders of magnitude observed
in the maximum likelihood estimates of the $b_{0ij}$, $\beta_{1j}$ and
residuals.

As in Section \ref{sec:plugin} we can robustify this likelihood in two
different ways: either against the distributional assumptions on the $\epsilon_{ijk}$ or on the
$b_{0ij}$. In the latter case we form the density estimate
\[
g_n(b;\betabold) = \frac{1}{-2nc_n}\sum_{i=1}^n \sum_{j \in \{a,f\}} K\left(
\frac{b-b_{0ij}+\beta_{0j}}{c_n} \right)
\]
and replace the second term in \eqref{randintercept.cdll} with
$-2nD(g_n(\cdot;\betabold),\phi_{0,\tau_0^2}(\cdot))$. Here we have used
\[
\betabold =
(\beta_{a0},\beta_{f0},\beta_{a1},\beta_{f1},b_{0a1},b_{0f1},\ldots,b_{0an},b_{0fn})
\]
as an argument to $g_n$ to indicate its dependence on the estimated parameters.
We have chosen to combine the $b_{0ai}$ and the $b_{0fi}$ together in order to
obtain the best estimate of $g_n$, rather than using a sum of disparities, one
for American and one for foreign students.

To robustify the residual distribution, we observe that we cannot replace the
first term with a single disparity based on the density of the combined
$\epsilon_{ijk}$ since the $b_{0ij}$ cannot be identified marginally.
Instead, we estimate a density at each $ij$:
\[
g_{ij,n}^{(c)}(e;\betabold) = \frac{1}{4nc_n} \sum_{k=1}^4 K \left( \frac{e - (Y_{ijk}
- b_{0ij} - \beta_{1j}t_k )}{c_n} \right)
\]
and replace the first term with $-\sum_{i=1}^n \sum_{j \in \{a,f\}} 4
D(g_{ij,n}^{(c)}(\cdot;\beta),\phi_{0,\sigma^2}(\cdot))$. This is the conditional
form of the disparity. Note that this reduces us to four points for each
density estimate; the limit of what could reasonably be employed. Naturally,
both replacements can be made.

Throughout our analysis, we used Hellinger distance as a disparity; we also
centered the $t_k$, resulting in $b_{0ij}$ representing the expected salary
of student $ij$ at age 50.  Bandwidths were fixed within a Metropolis sampling
procedures. These were chosen by estimating the $\hat{b}_{0ij}$ and
$\hat{\beta}_{1j}$ via least squares, and using these to estimate residuals and all
other parameters:
\begin{align*}
\hat{\beta}_{0j} & = \frac{1}{n} \hat{b}_{0i}, &
e_{ijk}  & = Y_{ijk} - \hat{b}_{0ij} - \hat{\beta}_{1j} t_k, \\
\hat{\sigma}^2 & = \frac{1}{8n-1} \sum_{ijk} e^2_{ijk}, &
\hat{\tau}_0^2 & = \frac{1}{2n-1} \sum_{ij} (\hat{b}_{0ij} - \hat{\beta}_{0j})^2.
\end{align*}
The bandwidth selector in \citet{SheatherJones91} was applied to the
$\hat{b}_{0ij}-\hat{b}_{0j}$ to obtain a bandwidth for $g_n(b;\betabold)$. The
bandwidth for $g_{ij,n}^{(c)}(e;\betabold)$ was chosen as the average of the
bandwidths selected for the $e_{ijk}$ for each $i$ and $j$. For each analysis,
a Metropolis algorithm was run for 200,000 steps and every 100th sample was
taken from the second half of the resulting Markov chain.  The results of this
analysis can be seen in Figure \ref{surveyintercept.fig}. Here we have plotted
only the differences $\beta_{f0}-\beta_{a0}$ and $\beta_{f1}-\beta_{a1}$ along
with the variance components. We observe that for posteriors that have not
robustified the random effect distribution, there appears to be a significant
difference in the rate of increase in income ($P(\beta_{f1} <\beta_{a1}) <
0.02$ for both posterior and replacing the observation likelihood with
Hellinger distance), however when the random effect likelihood is replaced with
Hellinger distance, the difference is no longer significant ($P(\beta_{f1}
<\beta_{a1}) > 0.145$ in both cases). We also observe that the estimated
observation variance for the model is significantly reduced for posteriors in
which the observation likelihood is replaced by Hellinger distance, but that
uncertainty in the difference $\beta_{f0}-\beta_{a0}$ is increased.

Investigating these differences, there were two foreign students who's over-all
expected rate of increase is negative and separated from the least-squares
slopes for all the other students. Removing these students increased the
posterior probability of $\beta_{a1} > \beta_{f1}$ to 0.11 and decreased the
estimate of $\sigma$ from 0.4 to 0.3. Removing the evident high outlier  with a
considerable departure from trend at age 45 in Figure \ref{surveydata.fig} in
Online Appendix \ref{sec:data} further reduced the EAP of $\sigma$ to 0.185, in
the same range as those obtained from robustifying the observation
distribution.

Further model exploration is possible. Online Appendix \ref{sec:randomslope} explores the use of a random slope model with additional modeling techniques, where a distinction in {\em average} slope between American and foreign students does not appear significant when the slope distribution is robustified via Hellinger distance.

\begin{figure}
\begin{center}
\begin{tabular}{cc}
\includegraphics[height=6.5cm,angle=0]{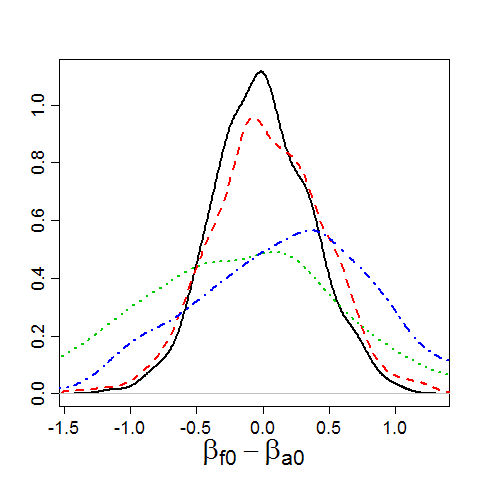}
&
\includegraphics[height=6.5cm,angle=0]{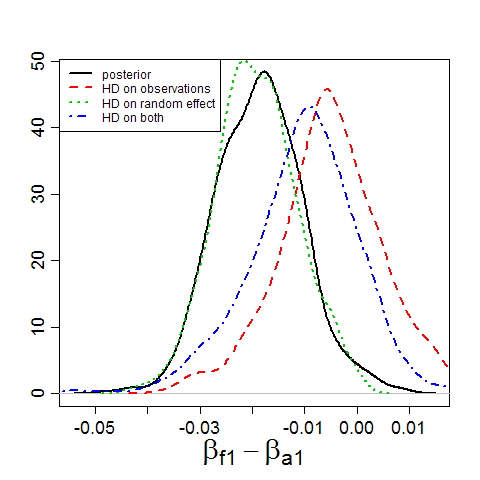}
\\
\includegraphics[height=6.5cm,angle=0]{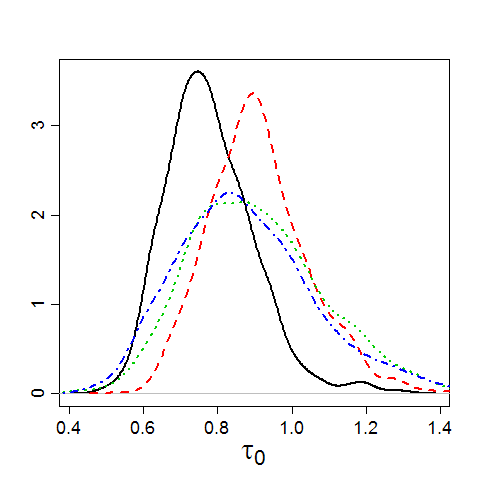}
&
\includegraphics[height=6.5cm,angle=0]{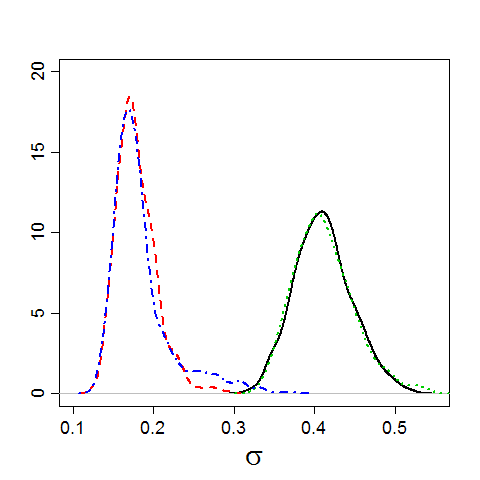}
\end{tabular}
\end{center}
\caption{Analysis of the class survey data using a random intercept model with
Hellinger distance replacing the observation likelihood, the random effect
likelihood or both. Top: differences in intercepts between foreign and American
students (left) and differences in slopes (right). Bottom: random effect
variance (left) and observation variance (right). Models robustifying the
random effect distribution do not show a significant difference in the slope
parameters. Those robustifying the observation distribution estimate a
significantly smaller observation variance. } \label{surveyintercept.fig}
\end{figure}

\section{Conclusions}

This paper combines disparity methods with Bayesian analysis to provide robust
and efficient inference across a broad spectrum of models. In particular, these
methods allow the robustification of any portion of a model for which the
likelihood may be written as a product of distributions for i.i.d. random
variables. This can be done without the need to modify either the assumed
data-generating distribution or the prior. In our experience, Metropolis
algorithms developed for the parametric model can be used directly to evaluate
the D-posterior and generally incur a modest increase in the acceptance rate
and computational cost. Our use of Metropolis algorithms in this context is
{\em deliberately naive} in order to demonstrate the immediate applicability of our
methods in combination with existing computational tools. We expect that a careful
study of the properties of these
methods will yield considerable improvements in both computational and sampling
efficiency.

The methods in this paper can be employed as a tool for model diagnostics; differences in
results by an application of posterior and D-posterior can indicate problematic
components of a hierarchical model. Further, estimated densities can indicate
how the current model may be improved \ghadd{by, for example, the addition of mixture  components}.
However, the D-posterior can also be
used directly to provide robust inference in an automated form.

Our mathematical results are given solely for i.i.d. data; ideas from
\citet{HookerVidyashankar09b} can be used to extend these
to the regression framework. Our proposal of hierarchical models remains under
mathematical investigation, but we expect that similar results can be
established in this case. The methodology can also be applied within a
frequentist context to define an alternative marginal likelihood for random
effects models, although the numerical estimation of such models is likely to
be problematic. Within this context, the choice of bandwidth $c_n$ can become
difficult. We have employed initial least-squares estimates above, but robust
estimators could also be used instead. Empirically, we have found our results to be
relatively insensitive
to the choice of bandwidth.

An opportunity for further development of the proposed methodology lies in
removing the boundedness of
many disparities in common use. These yield EDAP estimates with finite-sample breakdown
points of 1, indicating hyper-insensitivity to outliers. Theoretically, some
form of boundedness in $G$ has been used within proofs of the robustness of minimum
disparity estimators. \ghadd{However, transformations of $D(g_n,f_\ta)$ can yield non-bounded replacements for the log likelihood which retain both robustness and efficiency properties and this suggests an investigation of the relationship between appropriate transformations and the structure of the parameter space $f_\ta$.}


The use of a kernel density estimate may also be regarded as inconsistent with
a Bayesian context and it may be desirable to employ non-parametric
Bayesian density estimates as an alternative.  Results for disparity estimation
are dependent on properties of kernel density estimates and this
extension will require significant mathematical development.

There is considerable scope to extend these methods to further problems.
Robustification of the innovation distribution in time-series models, for
example, can be readily carried though through disparities and the hierarchical
approach will extend this to either the observation or the innovation process
in state-space models. The extension to continuous-time models such as
stochastic differential equations, however, remains an open and interesting
problem. More challenging questions arise in spatial statistics in which
dependence decays over some domain and where a collection of i.i.d. random
variables may not be available. There are also open questions in the
application of these techniques to non-parametric smoothing, and in functional
data analysis.

\subsection*{Acknowledgements}

Giles Hooker's research was supported by NSF grant
DEB-0813734 and the Cornell University Agricultural Experiment Station federal
formula funds Project No. 150446. Anand Vidyashankar's research was supported
in part by a grant from NSF DMS 000-03-07057 and also by grants from the
NDCHealth Corporation.

\appendix

\section{Proofs of Robustness} \label{robust.proofs}

\subsection{Proof of Theorem \ref{EDAP-MDE.thm}}

Before giving the proof we remark that it follows the lines of asymptotic
expansions for posterior distributions as outlined in, for example,
\citet{GDS06}. While we have provided explicit expressions only for the first
term of the expansion, further terms can be given analytically.

\begin{proof}
Let $\hat{\ta}(h)$ be the MDE for the density $h$.
Using a Taylor expansion of the prior density we have:
\begin{align*}
\pi\left(\hat{\ta}(h)+t/\sqrt{n}\right) & =
\pi\left(\hat{\ta}(h)\right)\left[1+n^{-1/2} t' \frac{\nabla_\ta
\pi(\hat{\ta}(h))}{\pi(\hat{\ta}(h))} + \frac{1}{2} n^{-1}t' \frac{\nabla_\ta^2
\pi(\hat{\ta}(h))}{\pi(\hat{\ta}(h))} t + o_p\left(n^{-1}\right) \right] \\
& = \pi\left(\hat{\ta}(h)\right)\left[1+n^{-1/2} t' b_1 + \frac{1}{2} n^{-1}t'
b_2 t + o_p\left(n^{-1}\right)  \right].
\end{align*}
Also, from the corresponding expansion of the disparity
\begin{align*}
& nD\left(h,f_{\hat{\ta}(h)+t/\sqrt{n}}\right) -
nD\left(h,f_{\hat{\ta}(h)}\right) \\
& \hspace{3cm} = \frac{1}{2}t' I^D(\hat{\ta}(h)) t + \frac{n^{-1/2}}{6}
\sum_{i,j,k} t_i t_j t_k   a_{3,ijk} + \frac{n^{-1}}{24} \sum_{ijkl}t_i t_j t_k
t_l a_{4,ijkl} + o_p\left(n^{-1}\right),
\end{align*}
 yielding
\begin{align*}
& \pi\left(\hat{\ta}(h)+t/\sqrt{n}\right)
e^{-nD\left(h,f_{\hat{\ta}(h)+t/\sqrt{n}}\right) +
nD\left(h,f_{\hat{\ta}(h)}\right)} \\ & \hspace{3cm}  = \pi\left(\hat{\ta}(h)\right)e^{- t'I^D(\hat{\ta}(h))
t/2}\left[1 + \frac{c_1(t)}{n^{1/2}} + \frac{c_2(t)}{n} + o_p(n^{-1}) \right],
\end{align*}
where $c_1(t)  =   \sum_i t_i b_{1,i} + \frac{1}{6} \sum_{ijk} t_i t_j t_k
a_{3,ijk},$ and
\[
c_2(t)  = \sum_{ij} \frac{b_{2,ij}}{2}t_i t_j + \sum_{ijkl}
\left( \frac{a_{3,ijk} b_{1,l}}{6} + \frac{a_{4,ijkl}}{24} \right)t^4 +
\sum_{ijk} \frac{a_{3,ijk}^2}{72} t_i^2 t_j^2 t_k^2.
\]
Here $b_1$ and $b_2$ provide constants in the Taylor of the prior and $a_3$ and $a_4$ provide the
corresponding third and four-order constants in a Taylor expansion of the disparity.

In particular,for the $i$th component of the EDAP vector we have that
\begin{align}
& T_n(h)_i - \hat{\ta}(h)_i = \frac{\int \left(\hat{\ta}(h)+t_i/\sqrt{n}\right)
e^{-nD\left(h,f_{\hat{\ta}(h)+t/\sqrt{n}}\right)+
nD\left(h,f_{\hat{\ta}(h)}\right)} \pi\left(\hat{\ta}(h)+t/\sqrt{n}\right)
dt}{\int e^{-nD\left(h,f_{\hat{\ta}(h)+t/\sqrt{n}}\right)+
nD\left(h,f_{\hat{\ta}(h)}\right)} \pi\left(\hat{\ta}(h)+t/\sqrt{n}\right) dt}  \nonumber \\
& \hspace{0.5cm} = \frac{ \left(\frac{\left|I^D(\hat{\ta}(h))\right|}{2 \pi}
\right)^{p/2} \left[ \hat{\ta}(h)_i +
n^{-1}\left[I^D(h)^{-1}\right]_{ii}\left(\sum_j
\frac{a_{3,jji}\left[I^D(h)^{-1}\right]_{jj}}{2} + \frac{\nabla_\ta
\pi(\hat{\ta}(h))_i}{\pi(\hat{\ta}(h))}\right) \right] +
o_p\left(n^{-1}\right)}{ \left(\frac{\left|I^D(\hat{\ta}(h))\right|}{2 \pi}
\right)^{p/2}  + o_p\left(n^{-1}\right)} \nonumber \\
& \hspace{0.5cm} = \hat{\ta}(h)_i +
n^{-1}\left[I^D(h)^{-1}\right]_{ii}\left(\sum_j
\frac{a_{3,jji}\left[I^D(h)^{-1}\right]_{jj}}{2} + \frac{\nabla_\ta
\pi(\hat{\ta}(h))_i}{\pi(\hat{\ta}(h))}\right) + o_p\left(n^{-1}\right) \label{expansion.3rdterm}
\end{align}
Here the boundedness of the third derivatives of $\pi$ ensures the integrability of the $o_p(n)$ terms.
\end{proof}

\subsection{Lemmas \ref{orthogonality.lem} and \ref{identifiability.lem}}

The following lemmas are needed in the proof of Theorem \ref{ass.breakdown.thm} below.

\begin{lem} \label{orthogonality.lem}
Under the conditions of Theorem \ref{ass.breakdown.thm}, for any $\delta$ there are $\ta_1^*>0$, $\ta_2^* > 0$ and $z^*>0$ such that
\begin{equation} \label{dzass1}
\sup_{|\ta| > \ta_1^*} \sup_{z > z^*} |D(h_{\alpha,z},\ta) - D(\alpha t_z,ta) - (1-\alpha)G'(\infty)| < \delta
\end{equation}
and
\begin{equation} \label{dzass2}
\sup_{|\ta| < \ta_1^*} \sup_{z > z^*} |D(h_{\alpha,z},\ta) - D((1-\alpha)g,ta) - \alpha G'(\infty)| < \delta.
\end{equation}
The same results hold for $D$ given by Hellinger distance taking $G'(\infty) = 0$.
\end{lem}

\begin{proof}
The arguments employed here largely follow those of \citet[][Theorem 4.1]{ParkBasu04}. For notational convenience throughout the following we will use $C(g,f) = G(g/f -1)f$ and in particular we note that for $f \rightarrow 0$, $C(g,f) \rightarrow G'(\infty)g$. We will also write $G^* = \sup_t \max(|G(t),|G'(t)|)$.

For \eqref{dzass1}, define
\[
A_{\ta^*_1,z^*} = \left\{ x: g(x) < \max\left( \sup_{z > z^*} t_z(x), \sup_{|\ta| > \ta^*_1} f_\ta(x) \right) \right\}
\]
by assumptions \eqref{tgortho} and \eqref{fgortho} for any $\epsilon > 0$ we can find $z^*$ and $\ta^*_1$ large enough that  $\int_{A_{\ta^*_1,z^*}} g(x) dx < \epsilon$ and $\sup_{|\ta| > \ta^*_1} \int_{A_{\ta^*_1,z^*}^c} f_{\ta}(x) dx < \epsilon$ so that writing
\[
\int_{A_{\ta^*_1,z^*}} C(h_{\alpha,z}(x),f_\ta(x)) dx = \int_{A_{\ta^*_1,z^*}} C(\alpha t_z(x),f_\ta(x)) dx +  \int_{A_{\ta^*_1,z^*}}  G'\left( \frac{h_{\alpha,z}^*(x)}{f_\ta(x)}-1\right)g(x)dx
\]
for $h^*{\alpha,z}(x)$ between $\alpha t_z(x)$ and $h_{\alpha,z}(x)$ yeilds
\[
\sup_{z>z^*} \sup_{|\ta|>\ta^*_1} \left| \int_{A_{\ta^*_1,z^*}} C(h_{\alpha,z}(x),f_\ta(x)) dx - \int C(\alpha t_z(x),f_\ta(x)) dx \right| \leq  2 G^* \epsilon.
\]
Similarly, on $A_{z^*,\theta^*}^c$ taking a Taylor expansion around $f = 0$.
\[
\sup_{z>z^*} \sup_{|\ta|>\ta^*_1} \left| \int_{A_{\ta^*_1,z^*}^c} C(h_{\alpha,z}(x),f_\ta(x)) dx - (1-\alpha) G'(\infty) \right| \leq  2 G^* \epsilon.
\]
Choosing $\epsilon < \delta/5G^*$ then yields the result.

For \eqref{dzass2} we proceed in a similar manner and define
\[
B_{\ta_2^*,z^*} = \left\{ x:  \sup_{z > z^*} t_z(x) < \max\left( g(x), \sup_{|\ta| \leq \ta_2^*} f_\ta(x) \right) \right\}
\]
where we observe that by assumptions \eqref{tgortho} and \eqref{tfortho}, for any $\epsilon>0$ and keeping $\ta^*_2$ bounded we can find $z^*$ large enough that
$\sup_{z > z^*} \int_{B_{\ta_2^*,z^*}} t_z(x)dx < \epsilon$ and $\sup_{|\ta| < \ta_2^*} \int_{B_{\ta_2^*,z^*}^c} f_\ta(x) dx < \epsilon$. Reversing the roles of $t_z$ and $g$ in the argument above then yields the result.

In the case of Hellinger distance, from the Cauchy-Schwartz inequality we observe that for $z^*$ sufficently large
\[
\sup_{\|\ta\| \in\Theta} \sup_{z > z^*} \left| \int  \sqrt{h_{\alpha,z}(x) f_\ta(x)} dx  - \int \sqrt{(1-\alpha)g(x) f_\ta(x)} dx - \int \sqrt{\alpha t_z(x) f_\ta(x)} dx\right| < \delta
\]
and the result follows from writing Hellinger distance as $HD(f,g) = 2 - 2 \int \sqrt{f(x) g(x)} dx$.

\end{proof}

\begin{lem} \label{identifiability.lem}
For any $\alpha \leq 1/2$, if there exists $\delta > 0$ such that
\[
\inf_z \inf_{\ta \in \Theta} D(\alpha t_z,\ta) >  \inf_{\ta \in \Theta} D((1-\alpha)g,\ta) + \delta,
\]
we can find $z^*$,  such that for every $\eta < \delta$ there exists $\theta^* > 0$ with
\begin{equation} \label{uniform.identification}
\inf_{z \geq z^*} \inf_{ \|\ta\| \geq \theta^* } \left( D(h_{\alpha,z},\ta) - D(h_{\alpha,z},\hat{\ta}(h_{\alpha,z})) \right) > \eta.
\end{equation}
\end{lem}

\begin{proof}
From Lemma \ref{orthogonality.lem}, for any $\eta^*$, we can find $\ta^*$ and $z^*$ so that for all $z > z^*$ and $|\ta| > \ta^*$
\[
D(h_{\alpha,z},\ta) > D(\alpha t_z,\ta) + (1-\alpha)G'(\infty) - \eta^*
\]
and defining $\ta_g = \argmin_{\ta \in \Theta} D((1-\alpha)g,\ta)$
\begin{align*}
D(h_{\alpha,z},\ta_g) & < D((1-\alpha)g,\ta_g) + \alpha G'(\infty) + \eta^*
\end{align*}
Thus
\[
\inf_{|\ta| \geq \ta^*} D(h_{\alpha,z},\ta)  \geq D( h_{\alpha,z},\ta_g) + \delta - 2\eta^*
\]
and taking $\eta^* = \eta/2$ yields the result.
\end{proof}

\subsection{Proof of Theorem \ref{ass.breakdown.thm}} \label{ass.breakdown.appendix}

\begin{proof}
We first observe that if the breakdown point of $\hat{\ta}(\cdot)$ is greater than $\alpha$, we have
\[
\sup_z | \hat{\ta}(h_{\alpha,z})| < \infty.
\]
We observe that for multivariate $\ta$ it is sufficient to prove the result for each co-ordinate and without loss of generality we will take $\hat{\ta}(h_{\alpha,z}) > \hat{\ta}(g)$.

Taking $\ta^*$ and $\eta$ and $z^*$ as in Lemma \ref{identifiability.lem}, we now observe that we can find $\epsilon^*$ so that for $|\theta - \hat{\ta}(h_{\alpha,z})| < \epsilon^*$ for all $z>z^*$ such that
\[
D(h_{\alpha,z},\ta) < D(h_{\alpha,z},\hat{\ta}(h_{\alpha,z})) + \eta/2
\]
then
\begin{align*}
T_n(h_{\alpha,z}) & = \frac{\int \ta e^{-n D(h_{\alpha,z},\ta)} \pi(\ta) d \ta}{\int e^{-n D(h_{\alpha,z},\ta)} \pi(\ta) d \ta} \\
& \leq (\hat{\ta}(h_{\alpha,z})+\epsilon)P_D(\ta < \hat{\ta}(h_{\alpha,z})+\epsilon) + \frac{e^{-nD(h_{\alpha,z},\hat{\ta}(h_{\alpha,z}))-n\eta}\int_{\hat{\ta}(h_{\alpha,z})+\epsilon}^{\infty} \ta \pi(\ta) d \ta}{e^{-nD(h_{\alpha,z},\hat{\ta}(h_{\alpha,z}))-n\eta/2} \int^{\hat{\ta}(h_{\alpha,z})+\epsilon^*}_{\hat{\ta}(h_{\alpha,z})} \pi(\ta) d \ta} \\
& \leq \hat{\ta}(h_{\alpha,z}) + e^{-n \eta/2} K(\hat{\ta}(h_{\alpha,z})),
\end{align*}
where $K(x) = \left( \int^{x+\epsilon^*}_{x} \pi(\ta) d \ta \right)^{-1} \int_{x+\epsilon}^{\infty} \ta \pi(\ta) d \ta $.
Since $\hat{\ta}(h_{\alpha,z})$ is bounded if the breakdown point is greater than $\alpha$, we observe that we can take
\[
K^*  = \sup \{ K(\ta): \ta < \sup_{z > z^*} \hat{\ta}(h_{\alpha,z}) \} < \infty
\]
and observe that
\[
\mbox{IF}_{\alpha,n}(z) < \mbox{IF}_{\alpha,\infty}(z) + K^*.
\]
Conversely, if the breakdown point is less than $\alpha$, we can find $z_n$ such that
\[
\hat{\ta}(h_{\alpha,z_n}) \rightarrow \infty
\]
and from Theorem \ref{EDAP-MDE.thm} for any $\delta$ we can choose a subsequence $k_n$ such that
\[ | T_{k_n}(h_{\alpha,z_{n}}) - \hat{\ta}(h_{\alpha,z_{n}}) | < \delta, \ \forall n. \]

\end{proof}

\subsection{Proof of Theorem \ref{zinfinite.thm}}

\begin{proof}

We observe that for any finite $M$ we have
\[
\sup_{\|\ta\| \leq M} \int f_\ta(x) t_z(x) dx \rightarrow 0
\]
by the limiting orthogonality of $f_\ta$ and $t_z$, hence we can find $M_k \rightarrow \infty$, $z_k \rightarrow \infty$ so that
\[
\sup_{\|\ta\| \leq M_k} \int f_\ta(x) t_{z_k} (x) dx \rightarrow 0
\]
and $\int g(x) t_{z_k}(x) dx \rightarrow 0$ and $\int_{\|\ta\| > M_k}  \pi(\ta) d \ta \rightarrow 0$. Hence
\[
T_n(h_{\alpha,z_k}) = \frac{\int_{\|\ta\| \leq M_k} \ta e^{-nD(h_{\alpha,z_k},\ta)} \pi(\ta) d\ta}{\int_{\|\ta\| \leq M_k} e^{-nD(h_{\alpha,z_k},\ta)} \pi(\ta) d\ta} + o_k(1)
\]
where $o_k(1)$ indicates a limit with respect to $k \rightarrow \infty$ and upon observing that $\sup_{h,\theta} e^{-nD(h,\theta)}  = 1$.

Now applying Lemma \ref{orthogonality.lem}, we have
\[
D(h_{\alpha,z_k},\ta) =  D((1-\alpha)g,\ta) + \alpha G'(\infty) + o_k(1)
\]
Combining these yields
\begin{align*}
T_n(h_{\alpha,z}) & = \frac{\int_{\|\ta\| \leq M_k} \ta e^{-nD((1-\alpha)g,\ta) -n\alpha G'(\infty)-no_k(1)} \pi(\ta) d\ta}{\int_{\|\ta\| \leq M_k} e^{-nD((1-\alpha)g,\ta) -n\alpha G'(\infty)-no_k(1)} \pi(\ta) d\ta} + o_k(1) \\
& = T_n( (1-\alpha)g ) + o_k(1)
\end{align*}
and the theorem follows.
\end{proof}

\subsection{Proof of Corollary \ref{breakdown.theorem}}

\begin{proof}
Under the assumptions, $\sup_{\ta,g} D(g,f_\ta) = R <\infty$ and $\inf_{\ta,g} D(g,f_\ta) = r > -\infty$.
Now let $h_{z,\alpha} = (1-\alpha)g + \alpha t_z$, then for all $\ta \in \Theta$,
$
e^{-n R} \leq e^{-n D(h_{z,\alpha},f_\ta)} < e^{-n r}, \hspace{0.2cm} \forall z,
\hspace{0.2cm} \forall \alpha \in [0 \ 1]
$
and therefore
\[
e^{n(r-R)} E_{\pi} \ta  = \frac{\int \ta e^{-n R} \pi(\ta)
d\ta}{ \int e^{-n r} \pi(\ta) d\ta } \leq E_{P_D(\ta|h_{z,\alpha})} \ta
\leq \frac{\int \ta e^{-n r} \pi(\ta) d\ta}{ \int e^{-n R} \pi(\ta) d\ta } =
e^{n(R-r)} E_{\pi} \ta.
\]
Where $E_{\pi}(\cdot)$ represents expectation with respect to the prior.
\end{proof}

\subsection{Proof of Theorem \ref{influence.theorem}}

\begin{proof}
It is sufficient to show that $\left| E_{P_D(\ta|g)} C_{z}(\ta,g) \right|
<\infty$ and $\left|E_{P_D(\ta|g)}\left[  \ta  C_{z}(\ta,g) \right]\right|  <
\infty. $ We will prove the first of these and the second will follow analogously.
\begin{align}
\left| E_{P_D(\ta|g)} C_{nz}(\ta,g) \right| &
\leq e^{n(R-r)} \left|\int C_{z}(\ta,g) \pi(\ta) d\ta\right| \nonumber \\
& \leq e^{n(R-r)}  \int \left| (g(x)-t_z(x)) \int
G'\left(\frac{g(x)}{f_\ta(x)}-1\right) \pi(\ta) d\ta \right|  dx \nonumber \\
& \leq e^{n(R-r)}  e_0 \int |g(x)-t_z(x)| dx \label{fromassumption}
\end{align}
where $\sup_{\ta,g} D(g,f_\ta) = R <\infty$ and $\inf_{\ta,g} D(g,f_\ta) = r >
-\infty$ and \eqref{fromassumption} follows from the assumption
\eqref{influence1} and the bound $\int |g(x)-t_z(x)| dx \leq 2$.
\end{proof}

Since $t_z(x)$ can be made to concentrate on regions where $r(x)$ is large, we
conjecture that the conditions in Theorem \ref{influence.theorem} are
necessary. In fact, this requirement means that the H-posterior influence
function will not be bounded for a large collection of parametric families.

\section{Proofs of Efficiency} \label{efficiency.proofs}

\subsection{Proof of Theorem \ref{l1convergence}}

We begin with the following Lemma:

\begin{lem} \label{wlemma}
Let
\[
w_n(t) = \pi(\hat{\ta}_n +
t/\sqrt{n})e^{-nD\left(g_n,f_{\hat{\ta}_n+t\sqrt{n}}\right) + n
D\left(g_n,f_{\hat{\ta}_n}\right)} - \pi(\ta_g) e^{-\frac{1}{2}t' I^D(\ta_g)t}
\]
then under A\ref{G.ass}-A\ref{MDE.efficiency}
\[
\int |w_n(t)| dt \convas 0 \mbox{ and }  \int \|t\|_{2} |w_n(t)| dt \convas 0.
\]
\end{lem}
\begin{proof}
We divide the integral into $A_1 = \{ \|t\|_{2} > \delta \sqrt{n}\}$ and $A_2 =
\{ \|t\|_{2} \leq \delta \sqrt{n}\}$:
\begin{equation} \label{wintegral}
\int |w_n(t)|dt = \int_{A_1} |w_n(t)| dt + \int_{A_2} |w_n(t)|dt
\end{equation}
and show that each vanishes in turn. First, since
$
\sup_{\ta \in \Ta} \left| D\left(g_n,f_{\ta}\right) - D\left(g,f_{\ta}\right) \right| \convas 0,
$
for some $\epsilon>0$ with
probability 1 it follows that by Assumption A\ref{identifiability},
\[
\exists N: \forall n \geq N,  \ \sup_{\|t\|_{2} > \delta}
D(g_n,f_{\hat{\ta}_n+t\sqrt{n}}) - D\left(g_n,f_{\hat{\ta}_n}\right) >
-\epsilon.
\]
This now allows us to demonstrate the convergence of the first term in
\eqref{wintegral}:
\begin{align}
\int_{A_1} |w_n(t)| dt & \leq  \int_{A_1} \pi(\hat{\ta}_n +
t/\sqrt{n})e^{-nD\left(g_n,f_{\hat{\ta}_n+t\sqrt{n}}\right) + n
D\left(g_n,f_{\hat{\ta}_n}\right)}dt \nonumber \\
& \hspace{1cm} + \int_{A_1} \pi(\ta_g)
e^{-\frac{1}{2}t' I^D(\ta_g) t}  dt \nonumber \\
& \leq  e^{-n \epsilon} + \pi(\ta_g)\left( \frac{|I^D(\ta_g)|}{2\pi}
\right)^{p/2} P(\|Z\|_2 > \sqrt{n} \delta) \label{wnas0}
\end{align}
where $Z$ is a $N(0,I^D(\ta_g))$ random variable and \eqref{wnas0} converges to
zero almost surely.

We now deal with the second term in \eqref{wintegral}. Notice that
\[
nD\left(g_n,f_{\hat{\ta}_n + t/\sqrt{n}}\right) -
nD\left(g_n,f_{\hat{\ta}_n}\right) = \frac{1}{2} t' I^D_n(\ta'_n)
\]
for $\ta'_n = \hat{\ta}_n + \gamma t/\sqrt{n}$ with $0\leq \gamma \leq 1$
and therefore
\[
w_n(t)  = \pi(\hat{\ta}_n + t/\sqrt{n}) e^{-\frac{1}{2} t' I^D_n(\ta'_n)} -
\pi(\ta_g) e^{-\frac{1}{2}t' I^D(\ta_g)t} \rightarrow 0
\]
for every $t$.

By Assumption
A\ref{ass.smooth} we can choose $\delta$ so that $I^D(\ta) \succ 2M$ if
$\|\ta-\ta_g\|_2 \leq 2\delta$ for some positive definite matrix $M$ where $A \succ B$ indicates $t'At
> t'Bt$ for all $t$.
Since $\|\ta'_n - \hat{\ta}_n\| \leq \delta$ with probability 1 for all $n$
sufficiently large $ \exp \left(-nD\left(g_n,f_{\hat{\ta}_n+t\sqrt{n}}\right) +
n D\left(g_n,f_{\hat{\ta}_n}\right)\right) \leq \exp\left(-\frac{1}{2} t'M
t\right). $ Therefore
\[
\int_{A_2} |w_n(t) | dt \leq \int_{A_2} \pi(\hat{\ta}_n +
t/\sqrt{n})e^{-\frac{1}{2} t' M t} + \pi(\ta_g) \int_{A_2}
e^{-\frac{1}{2} t I^D(\ta_g) t} dt < \infty.
\]
and the result follows from the pointwise convergence of $w(t)$ and the dominated convergence theorem.

We can prove
$
\int \|t\|_{2} |w_n(t)| dt \convas 0
$
in an analogous manner by observing that on $A_1$
\begin{align*}
\int_{A_1} \| t \|_2 |w_n(t)| dt & \leq  \int_{A_1} \|t\|_2 \pi(\hat{\ta}_n +
t/\sqrt{n})e^{-nD\left(g_n,f_{\hat{\ta}_n+t\sqrt{n}}\right) + n
D\left(g_n,f_{\hat{\ta}_n}\right)}dt \\
& \hspace{1cm} + \int_{A_1} \pi(\ta_g) \|t\|_2 e^{-\frac{1}{2}t' I^D(\ta_g) t}
dt
\end{align*}
and on $A_2$, $\|t\|_2 |w_n(t)| \convas 0$ and
\[
\int_{A_2} \|t\|_2 |w_n(t) | dt \leq \int_{A_2} \|t\|_2 \pi(\hat{\ta}_n +
t/\sqrt{n})e^{-\frac{1}{2} t' M t} + \pi(\ta_g) \int_{A_2}
\|t\|_2 e^{-\frac{1}{2} t I^D(\ta_g) t} dt < \infty.
\]
\end{proof}

Using this lemma, we prove Theorem \ref{l1convergence}.

\begin{proof}
First, from Assumption A\ref{MDE.efficiency},
$
\sqrt{n} \left( \hat{\ta}_n - \ta_g \right) \convd N(0,I^D(\ta_g)),
$
using that $\int|g_n(t)-f_{\ta_g}(t)|dt \convas 0$, and Assumption A\ref{l1continuity}, it follows that
\[
\sup_{\ta \in \Ta} \left| D\left(g_n,f_{\ta}\right) - D\left(g,f_{\ta}\right) \right| \convas 0
\]
and since $\hat{\ta}_n \convas \ta_g$ and Assumption A\ref{ass.smooth}
\[
D\left(g_n,f_{\hat{\ta}_n}\right) \convas D\left(g,f_{\ta_g}\right), \
\nabla_\ta D\left(g_n,f_{\hat{\ta}_n}\right) \convas
\nabla_\ta D\left(g,f_{\ta_g}\right), \ \nabla^2_\ta
D\left(g_n,f_{\hat{\ta}_n}\right) \convas \nabla^2_\ta
D\left(g,f_{\ta_g}\right).
\]

Now, we write that $ \pi^{*D}_n(t)  = \kappa_n^{-1} \pi(\hat{\ta}_n +
t/\sqrt{n})\exp\left[-nD\left(g_n,f_{\hat{\ta}_n+t\sqrt{n}}\right) + n
D(g_n,f_{\hat{\ta}})\right] $ where $\kappa_n$ is chosen so that $\int \pi^{*D}_n(t)
dt = 1$. Let
\[
w_n(t) = \pi(\hat{\ta}_n +
t/\sqrt{n})e^{-nD\left(g_n,f_{\hat{\ta}_n+t\sqrt{n}}\right) + n
D\left(g_n,f_{\hat{\ta}_n}\right)} - \pi(\ta_g) e^{-\frac{1}{2}t' I^D(\ta_g)t}.
\]
From Lemma \ref{wlemma}, it follows that
$
\int |w_n(t)| dt \convas 0
$
from which
\begin{align*}
\kappa_n = \int \pi(\hat{\ta}_n +
t/\sqrt{n})e^{-nD\left(g_n,f_{\hat{\ta}_n+t\sqrt{n}}\right) + n
D\left(g_n,f_{\hat{\ta}_n}\right)} dt & \convas \pi(\ta_g) \int
e^{-\frac{1}{2}t' I^D(\ta_g) t} dt \\
&  = \pi(\ta_g) \left( \frac{2\pi}{|I^D(\ta_g)|}
\right)^{p/2}
\end{align*}
and
\begin{align*}
& \lim_{n \rightarrow \infty} \int \left| \pi^{*D}_n(t) -
\left(\frac{I^D(\ta_g)}{2 \pi}\right)^{p/2} e^{-\frac{1}{2}t' I^D(\ta_g) t}
\right| dt \\
& \hspace{2.5cm} \leq \kappa_n^{-1} \int |w_n(t)| dt +  \left(
\frac{2\pi}{|I^D(\ta_g)|} \right)^{p/2} \left|
\kappa_n^{-1}\pi(\ta_g) - \left( \frac{|I^D(\ta_g)|}{2\pi} \right)^{p/2} \right| \\
& \hspace{2.5cm} \convas 0.
\end{align*}
That the result holds for $I^D(\ta_g)$ replaced with $\hat{I}^D_n(\hat{\ta}_n)$ follows
from the almost sure convergence of the latter to the former.
\end{proof}

\subsection{Proof of Theorem \ref{posteriorefficiency}}

\begin{proof}
Let $t = (t_1,\ldots,t_p)$, from Theorem \ref{l1convergence}
\[
\int t_i \pi^{*D}(t|x_1,\ldots,x_n) \convas \left(\frac{2 \pi}{|I^D(\ta_g)|}
\right)^{p/2} \int t_i e^{-\frac{1}{2} t' I^D(\ta_g) t} dt = 0.
\]
Since
$
\ta^*_n = E(\hat{\ta}_n + t/\sqrt{n}|X_1,\ldots,X_n)
$
we have
\[
\sqrt{n}\left( \ta^*_n - \hat{\ta}_n \right) \convas \left(\frac{2
\pi}{|I^D(\ta_g)|} \right)^{p/2} \int t e^{-\frac{1}{2} t' I^D(\ta_g) t} dt =
0.
\]
Since $\sqrt{n}\left(\hat{\ta}_n - \ta_g \right) \convd
N\left(0,I^D(\ta_g)\right)$, it follows that  $\sqrt{n}\left(\ta^*_n - \ta_g
\right) \convd N\left(0,I^D(\ta_g)\right)$; hence
$\ta^*_n$ is asymptotically normal, efficient as well as robust.
\end{proof}

\subsection{Efficiency Conditions for Minimum Disparity Estimators}
\label{sec:efficiency.conditions}

Here we provide conditions that ensure the consistency and asymptotic normailty
of the minimum-disparity estimator $\hat{\theta}_n$. There is a slight
variation through the literature in conditions required for efficiency (see
\citet{Beran77}, \citet{BSV97}, \citet{ParkBasu04} and
\citet{ChenVidyashankar06}). The conditions given below are adapted from
\citet{ChenVidyashankar06} for the specific case of Hellinger distance. A small
modification of these conditions will also provide the consistency and
asymptotic normality for more general disparities under appropriate conditions
on $G(\cdot)$. These are in addition to conditions A\ref{G.ass}-A\ref{identifiability}.

We first require conditions on the data and the proposed parametric family:
\begin{enumerate}[(D1)]
\item $X_1,\ldots,X_n$ are i.i.d. with distribution given by the density function $g(\cdot)$.

\item The parameter space $\Theta$ is locally compact.

\item $f_\ta(\cdot)$ is twice continuously differentiable with respect to $\ta$.

\item $\left\| \nabla_\ta \left(\sqrt{f_\ta(\cdot)} \right)\right\|_2$ is continuous and bounded.

\item $f^{-1}_\ta(x) \left[\nabla_\ta f_\ta(x)\right]$ is continuous and bounded in $L_2$ at $\ta = \ta_g$.

\item $f^{-1/2}_\ta(x) \left[\nabla^2_\ta f_\ta(x)\right] -
f^{-3/2}_\ta(x) \left[\nabla_\ta f_\ta(x)\right]  \left[\nabla_\ta
f_\ta(x)\right]^T$ is continuous and bounded in $L_2$ at $\ta = \ta_g$.

\item $f^{-1/2}_\ta(x) \left[\nabla_\ta f_\ta(x)\right] \left[\nabla_\ta f_\ta(x)\right]^T$
is continuous and bounded in $L_2$ at $\ta = \ta_g$.
\end{enumerate}

We also require conditions on the kernel function $K$ in the kernel density
estimate and its relationship to the parametric density family:
\begin{enumerate}[(K1)]
\item $K(\cdot)$ is symmetric about 0 and $\int K(t)dt = 1$.

\item The bandwidth is chosen so that $c_n \rightarrow 0$, $n c_n^2 \rightarrow 0$, $n c_n \rightarrow \infty$.

\item There is a sequence $a_n$, $n \geq 0$ diverging to infinity such that
\begin{enumerate}
\item For $X$ a random variable with density $f_{\ta_g}(\cdot)$
\[
\lim_{n \rightarrow \infty} n  P\left( |X| > a_n \right) = 0,
\]
\item
\[
\sup_{n \geq 1} \sup_{|x| < a_n} \sup_{t \in R} \left|K(t) \frac{f_{\ta_g}(x +
t c_n)}{f_{\ta_g}(x)} \right| < \infty,
\]
\item The parametric score functions have regular central behavior relative to the bandwidth:
\[
\lim_{n \rightarrow \infty} \frac{1}{n^{1/2} c_n} \int_{-a_n}^{a_n}
\frac{\nabla_\ta f_{\ta_g}(x)}{f_{\ta_g}(x)} dx = 0
\]
and
\[
\lim_{n \rightarrow \infty} \frac{1}{n^{1/2} c_n^4} \int_{-a_n}^{a_n}
\frac{\nabla_\ta f_{\ta_g}(x)}{f_{\ta_g}(x)} dx = 0,
\]
\item The score functions are smooth with respect to $K$ in an $L_2$ sense:
\[
\lim_{n \rightarrow \infty} \sup_{t \in R} \int_{\mathbb{R}} K(t) \left[ \frac{
\nabla_{\ta} f_{\ta_g}(x+c_nt)}{f_{\ta_g}(x+c_n t)} - \frac{\nabla_\ta
f_{\ta_g}(x)}{f_{\ta_g}(x)} \right]^2 f_{\ta_g}(x) dx = 0.
\]
\end{enumerate}
\end{enumerate}

This statement of assumptions in particular remove the condition that $K(t)$
has compact support, which was assumed in \citet{Beran77,BSV97,ParkBasu04}.
These assumptions significantly expand the class of kernels available for use
and hence expands the applicability of Theorem \ref{l1convergence} (see
\citet{HookerVidyashankar09} for formal details). In practice, it is
numerically more stable to use a Gaussian kernel or some other distribution
with support on the whole real line and we have used Gaussian kernels
throughout our numerical experiments. We also follow \citet{ChenVidyashankar06} in removing the
assumption of compactness of $\Theta$, replacing it with local compactness plus the identifiability
condition A\ref{identifiability}.

\section{Reducing the Kernel Density Dimension} \label{sec:reduction}

The conditional disparity formulation outlined above requires the estimation of
the density of a response conditioned on a potentially high-dimensional set of
covariates; this can result in asymptotic bias and poor performance in small
samples. In this section, we explore two methods for reducing the dimension of
the conditioning spaces.  The first is referred to as the
``marginal formulation'' and requires only a univariate, unconditional, density
estimate. This is a Bayesian extension of the approach suggested in
\citet{HookerVidyashankar09}.  It is more stable and computationally efficient
than schemes based on nonparametric estimates of conditional densities.
However, in a linear-Gaussian model with Gaussian covariates, it requires an
external specification of variance parameters for identifiability. For this
reason, we propose a two-step Bayesian estimate. The asymptotic analysis for
i.i.d. data can be extended to this approach by using the technical ideas in
\citet{HookerVidyashankar09}.

The second method produces a conditional formulation that relies on the
structure of a homoscedastic location-scale family $P(Y_i|X_i,\ta,\sigma) =
f_\sigma(y-\eta(X_i,\ta))$ and we refer to it as the
``conditional-homoscedastic'' approach. This method provides a full conditional
estimate by replacing a non-parametric conditional density estimate with a
two-step procedure as proposed in \citet{Hansen04}. The method involves first
estimating the mean function non-parametrically and then estimating a density
from the resulting residuals.

\subsection{Marginal Formulation}

\citet{HookerVidyashankar09} proposed basing inference on a marginal estimation
of residual density in a nonlinear regression problem. A model of the form
\[ Y_i = \eta(X_i,\ta) + \epsilon_i \]
is assumed for independent $\epsilon_i$ from a scale family with mean zero and
variance $\sigma^2$. $\ta$ is an unknown parameter vector of interest. A
disparity method was proposed based on a density estimate of the residuals
\[ e_i(\ta) = Y_i - \eta(X_i,\ta) \]
yielding the kernel estimate
\begin{equation} \label{marginal.density}
g_n^{(m)}(e,\ta,\sigma) = \frac{1}{n c_n} \sum K\left( \frac{ e - e_i(\ta)/\sigma}{
c_n} \right)
\end{equation}
and $\ta$ was estimated by minimizing
$D(\phi_{0,1}(\cdot),g_n^{(m)}(\cdot,\ta,\sigma))$ where $\phi_{0,1}$ is is the
postulated density. As described above, in a Bayesian context we replace the
log likelihood by \\ $-n D(\phi_{0,1}(\cdot),g_n^{(m)}(\cdot,\ta,\sigma))$.  Here we
note that although the estimated of $g_n^{(m)}(\cdot,\ta,\sigma)$ need not have
zero mean, it is compared to the centered density $\phi_{0,1}(\cdot)$ which
penalizes parameters for which the residuals are not centered.

This formulation has the advantage of only requiring the estimate of a
univariate, unconditional density $g_n^{(m)}(\cdot,\ta,\sigma)$. This reduces the
computational cost considerably as well as providing a density estimate that is
more stable in small samples.

\citet{HookerVidyashankar09} proposed a two-step procedure to avoid
identifiability problems in a frequentist context.  This involves replacing
$\sigma$ by a robust external estimate $\tilde{\sigma}$. It was observed that
estimates of $\ta$ were insensitive to the choice of $\tilde{\sigma}$. After an
estimate $\hat{\ta}$ was obtained by minimizing
$D(\phi_{0,1}(\cdot),g_n^{(m)}(\cdot,\ta,\tilde{\sigma}))$, an efficient estimate
of $\sigma$ was obtained by re-estimating $\sigma$ based on a disparity for the
residuals $e_i(\hat{\ta})$. A similar process can be undertaken here.

In a Bayesian context a plug-in estimate for $\sigma^2$ also allows the use of
the marginal formulation: an MCMC scheme is undertaken with the plug-in value
of $\sigma^2$ held fixed. A pseudo-posterior distribution for $\sigma$ can then
be obtained by plugging in an estimate for $\ta$ to a Disparity-Posterior for
$\sigma$. More explicitly, the following scheme can be undertaken:
\begin{enumerate}
\item Perform an MCMC sampling scheme for $\ta$ using a plug-in estimate for
$\sigma^2$.

\item Approximate the posterior distribution for $\sigma^2$ with an MCMC scheme
to sample from the D-posterior $P_D(\sigma^2|y) =
e^{-nD(g_n(\cdot,\hat{\ta},\sigma),\phi_{0,1}(\cdot))} \pi(\sigma^2)$ where
$\hat{\ta}$ is the EDAP estimate calculated above.
\end{enumerate}
This scheme is not fully Bayesian in the sense that fixed estimators of
$\sigma$ and $\ta$ are used in each step above. However, we conjecture that, as in
\citet{HookerVidyashankar09}, the two-step procedure will result in statistically efficient estimates
and asymptotically correct credible regions. We note that while we have
discussed this formulation with respect to regression problems, it can also be
employed with the plug-in procedure for random-effects models and we use this
in Section \ref{sec:survey}, below.

The formulation presented here resembles the methods proposed in
\citet{PakBasu98} based on a sum of disparities between  weighted density
estimates of the residuals and their expectations assuming the parametric
model. For particular combinations of kernels and densities, these estimates
are efficient, and the sum of disparities, appropriately scaled, should also be
substitutable for the likelihood in order to achieve an alternative
D-posterior.

\subsection{Nonparametric Conditional Densities for Regression Models in Location-Scale Families}

Under a homoscedastic location-scale model (where the errors are assumed to be i.i.d.) $p(Y_i|X_i,\ta,\sigma) =
f_{\sigma}(Y_i-\eta(X_i,\ta))$ where $f_\sigma$ is a distribution with zero
mean, an alternative density estimate may be used. We first define a
non-parametric estimate of the mean function
\[ m_n(x) = \frac{\sum Y_i K\left(\frac{x-X_i}{c_{n2}} \right) }{ \sum K\left(\frac{x-X_i}{c_{n2}} \right) } \]
and then a non-parametric estimate of the residual density
\[ g_n^{(c2)}(e) = \frac{1}{n c_{n1}} \sum K \left( \frac{e - y_i + m_n(X_i)}{c_{n2}} \right). \]
We then consider the disparity between the proposed $f_{\ta,\sigma}$ and $g_n$:
\[ D^{(c2)}(g_n,\ta,\sigma) = \sum D(g_n^{(c2)}(\cdot+m(X_i)),f_\sigma(\cdot+\eta(X_i,\ta))). \]
As before, $- D^{(c2)}(g_n,f)$ can be substituted for the log likelihood in an
MCMC scheme.

\citet{Hansen04} remarks that in the case of a homoscedastic conditional
density, $g_n^{(c2)}$ has smaller bias than $g_n^{(c)}$. This formulation does not
avoid the need to estimate the high-dimensional function $m_n(x)$. However, the
shift in mean does allow the method to escape the identification problems of
the marginal formulation while retaining some of its stabilization.

Online Appendix \ref{sim:linreg} gives details of a simulation study of both
conditional formulations and the marginal formulation above for a regression
problem with a three-dimensional covariate.  All disparity-based methods
perform similarly to using the posterior with the exception of the conditional
form in Section \ref{sec:conditional} when Hellinger distance is used which
demonstrates a substantial increase in variance. We speculate that this is due
to the sparsity of the data in high dimensions creating inliers; negative
exponential disparity is less sensitive to this problem \citep{BSV97}.

\section{Computational Considerations and Implementation} \label{computation}

Our experience is that the computational cost of employing Disparity-based
methods as proposed above is comparable to employing an MCMC scheme for the
equivalent likelihood and generally requires an increase in computation time by
a factor of between 2 and 10. Further, the comparative advantage of employing
estimates \eqref{mcestimate} versus \eqref{GHestimate} depends on the context
that is used.

Formally, we assume $N$ Monte Carlo samples is \eqref{mcestimate} and $M$
Gauss-Hermite quadrature points in \eqref{GHestimate} where typically $M < N$.
In this case, the cost of evaluating $g_n(z_i)$ in \eqref{mcestimate} is
$O(nN)$, but this may be pre-computed before employing MCMC, and the cost of
evaluating \eqref{mcestimate} for a new value of $\theta$ is $O(N)$. In
comparison, the use of \eqref{GHestimate} requires the evaluation of
$g_n(\xi_i(\theta))$ at each iteration at a $O(nM)$ each evaluation.

Within the context of conditional disparity metrics, we assume $N$ Monte Carlo
points used for each $X_i$ in the equivalent verion of \eqref{mcestimate} for
\eqref{cdisp} and note that in this context $N$ can be reduced due to the
additional averaging over the $X_i$. The cost of evaluating $g_n^{(c)}(z_j|X_i)$
from \eqref{conditional.kernel} for all $z_j$ and $X_i$ is $O(n^2N)$ for
\eqref{mcestimate} and $O(n^2M)$ for \eqref{GHestimate}. Here again the
computation can be carried out before MCMC is employed for \eqref{mcestimate},
requiring $O(nN)$ operations. In \eqref{GHestimate} the denominator of
\eqref{conditional.kernel} can be pre-computed, reducing the computational cost
of each iteration to $O(nM)$; however, in this case we will not necessarily
expect $M<N$.  Similar calculations apply to estimates based on $g_n^{(c2)}$.

For marginal disparities $g_n^{(m)}$ in \eqref{marginal.density} changes for each
$\theta$, requiring $O(nM)$ calculations to evaluate \eqref{GHestimate}.
Successful use of \eqref{mcestimate} would require the $z_i$ to vary smoothly
with $\theta$ and would also require the re-evaluation of $g_n^{(m)}(z_i)$ at a
cost of $O(nN)$ each iteration. Within the context of hierarchical models
above, $g_n$ varies with latent variables and this the use of
\eqref{mcestimate} will generally be more computationally efficient. The cost
of evaluating the likelihood is always $O(n)$.

While these calculations provide general guidelines to computational costs, the
relative efficiency of \eqref{mcestimate} and \eqref{GHestimate} strongly
depends on the implementation of the procedure. Our simulations have been
carried out in the \texttt{R} programming environment where we have found
\eqref{mcestimate} to be computationally cheaper anywhere it can be employed.
However, this will be platform-dependent -- changing with what routines are
given in pre-compiled code, for example -- and will also depend strongly on the
details of our implementation.

\section{Simulation Studies}

While we have established attractive asymptotic properties for these estimators their
finite sample properties remain an important source of investigation. Since these
estimates are based on nonparametric density estimation, we may suspect that they require
large sample sizes before their asymptotic properties become apparent. In fact, our studies
below demonstrate good performance even for small sample sizes.

\subsection{Gaussian and exponential-Gamma Distributions -- The i.i.d. Case} \label{sim:iid}

We undertook a simulation study for the normal mean example in Figure
\ref{likapprox} to examine the efficiency and robustness of Hellinger and
Negative-Exponential posterior samples.  1,000 sample data sets of size 20 from
a $N(5,1)$ population were generated. For each sample data set, a random walk
Metropolis algorithm was run for 20,000 steps using a $N(0,0.5)$ proposal
distribution and a $N(0,25)$ prior, placing the true mean one prior standard deviation above the prior mean. The kernel bandwidth was selected by the
bandwidth selection in \citet{SheatherJones91}. H- and N-posteriors were easily
calculated by combining the \texttt{KernSmooth} \citep{KernSmooth} and
\texttt{LearnBayes} \citep{LearnBayes} packages in \texttt{R}. Expected {\em a
posteriori} estimates for the sample mean were obtained along with 95\%
credible intervals from every second sample in the second half of the MCMC chain.
Outlier contamination was investigated by reducing the last one, two or five
elements in the data set by 3, 5 or 10. This choice was made so that both outliers and
prior influence the EDAP in the same direction. The analytic
posterior without the outliers is normal with mean 4.99 (equivalently, bias
of -0.01) and standard deviation 0.223. The results of this simulation are
summarized in Table \ref{simtable1}. As can be expected, the standard Bayesian
posterior suffers from sensitivity to large negative values whereas the disparity-based
methods remain nearly unchanged. Near-outliers at the smaller value of -3
resulted in similar biases across all methods. A comparison of CPU time
indicates that the use of disparity methods required a little more than twice
the computational effort as compared to using the likelihood within an MCMC
method.

In order to provide a comparison with classical robust methods, we conducted a comparison with an estimate based on a Huberized mean. Table \ref{hubsimtable1} provides the results of these estimates using the Huber loss as a log likelihood within an MCMC scheme and by the minimizer of the Huber loss. In both cases, we tried cut-points between quadratic and linear costs at the normal 0.8 and 0.99 critical values. We also considered replacing the log likelihood within the MCMC scheme with a loss corresponding to Tukey's biweight function with cut point 4.685; this has the appeal of being similar to disparity methods in both having a re-descending influence function and tails of the likelihood that do not go to zero at infinity. We notice here that Huberized methods with cut point at the 0.8 critical value and Tukey's biweight method both showed noticeable increases in standard deviation compared to Hellinger methods and the likelihood when no outliers were added to the data. The negative exponential disparity had a similar increase in variance in this case which we speculate is due to having relatively heavier tails than Hellinger distance and therefore greater influence from the prior (the maximal value of NED is $e^{-1}$ where as that of 2HD is 4).  While Huber both methods exhibited less bias due to outliers than the likelihood, the bias still grew substantially with large outliers. By contrast, both biweight and disparity methods demonstrate redescending influence as outliers are placed at more extreme values, but we observe that the biweight estimates were more strongly affected for larger numbers of contaminated observations.  Disparity-based methods had credible intervals between 5\% and 15\% wider than those based on the likelihood, but these were smaller than those for Huberized and biweight methods when a prior was employed.  The computing times for MCMC methods using Huber and biweight likelihoods are somewhat greater than for the use of the true likelihood, although computing the minimizing value is unsurprisingly faster.

The problem of estimating a Gaussian mean is made relatively straightforward by
the symmetry of the distribution. We therefore conducted a further study,
estimating both shape and scale parameters in an exponential-gamma distribution (ie, $\exp(X_i)$ has a $Gamma$ distribution, and the domain of the distribution is the whole real line).
In this
case, the shape parameter was chosen at 5 and the scale parameter at 0.25 and
these were given $\chi^2$ priors with degrees of freedom 3 and 0.3
respectively. 5,000 data sets were simulated of 20 points each and the
D-posteriors were calculated as above both with and without an outlier placed
at ln(20). For this chain a random walk Metropolis algorithm was again
conducted with the random walk on the log shape and log scale parameters again
using the \texttt{LearnBayes} package. Table \ref{simtable2} reports the
tabulated results. We note in particular that the efficiency of the H-posterior
has been considerably reduced, as have its coverage properties; additionally
two simulation runs were removed due to poor convergence. This is explained as
being due to the inlier effect; the skewness of the gamma distribution produces
density estimates $g_n$ that tend to have ``holes'' and the Hellinger disparity
is sensitive to these; an example is give in Figure \ref{gamma.holes}. By
contrast, the negative exponential disparity is much less sensitive.

\begin{table}
\begin{center}
{\Large No Outliers}\\
\vspace{0.2cm}

\begin{tabular}{|l|ccccc|}\hline
                     &   Bias   & SD    & Coverage & Length & CPU Time \\ \hline
Posterior            & -0.015   & 0.222 &   0.956  & 0.873  & 3.393 \\
Hellinger            & -0.015   & 0.225 &   0.954  & 0.920  & 7.669 \\
Negative Exponential & -0.018   & 0.229 &   0.973  & 1.022  & 7.731 \\
\hline
\end{tabular}
\vspace{0.5cm}

{\Large Outliers} \\
\vspace{0.2cm}

\begin{tabular}{|lc||ccc|ccc|ccc|} \hline
& & \multicolumn{3}{c}{1 Outlier} &  \multicolumn{3}{c}{2 Outliers} &
\multicolumn{3}{c}{5 Outliers} \\ \hline
&  Loc  & Bias    & SD    & Coverage & Bias   & SD    & Coverage & Bias   & SD    & Coverage \\ \hline \hline
Posterior
& -3     & -0.164  & 0.219  & 0.883   & -0.300 & 0.206 & 0.722    & -0.637 & 0.182 & 0.100 \\
& -5     & -0.264  & 0.219  & 0.778   & -0.490 & 0.206 & 0.375    & -1.053 & 0.182 & 0.001 \\
& -10    & -0.513  & 0.219  & 0.360   & -0.965 & 0.207 & 0.004    & -2.093 & 0.182 & 0.000 \\ \hline
HD
& -3     &  -0.109 & 0.246 & 0.920    & -0.194 & 0.275 & 0.859    & -0.237 & 0.299  & 0.770 \\
& -5     &  -0.027 & 0.238 & 0.942    & -0.040 & 0.257 & 0.928    & -0.024 & 0.305  & 0.865 \\
& -10    &  -0.014 & 0.234 & 0.948    & -0.019 & 0.249 & 0.935    &  0.018 & 0.286  & 0.883 \\ \hline
NED
& -3     &  -0.080 & 0.256 & 0.959    & -0.133 & 0.279 & 0.933    & -0.166 & 0.308  & 0.893 \\
& -5     &  -0.020 & 0.238 & 0.977    & -0.025 & 0.243 & 0.968    & -0.015 & 0.264  & 0.948 \\
& -10    &  -0.017 & 0.237 & 0.973    & -0.020 & 0.241 & 0.970    & -0.007 & 0.260  & 0.952 \\ \hline
\end{tabular}
\end{center}
\caption{A simulation study for a normal mean using the usual posterior, the
Hellinger posterior and the Negative Exponential posterior. Columns in the
first table give the bias and variance of the posterior mean, coverage and
average CPU time of the central 95\% credible interval based on 1,000
simulations. The second table provides results for 1, 2, and 5 outliers (large
columns) located at -3, -5, -10 (column Loc) for the posterior, Hellinger
distance (HD) and negative exponential disparity (NED).}  \label{simtable1}
\end{table}

\begin{table}
\begin{center}
{\Large No Outliers}\\
\vspace{0.2cm}

\begin{tabular}{|l|ccccc|}\hline
             & Bias   & SD    & Coverage & Length & CPU Time \\ \hline \hline
Bayes 80     & -0.017 & 0.229 &   0.973  & 0.978  & 5.337 \\
Bayes 99     & -0.015 & 0.223 &   0.956  & 0.877  & 4.962 \\ \hline
Freq 80      & -0.004 & 0.229 &   0.948  & 0.894  & 0.002 \\
Freq 99      & -0.004 & 0.223 &   0.960  & 0.894  & 0.001 \\ \hline
Biweight     & -0.017 & 0.228 &   0.977  & 1.007  & 3.523 \\ \hline
\end{tabular}
\vspace{0.5cm}

{\Large Outliers}\\
\vspace{0.2cm}

\begin{tabular}{|lc||ccc|ccc|ccc|} \hline
& & \multicolumn{3}{c}{1 Outlier} &  \multicolumn{3}{c}{2 Outliers} &
\multicolumn{3}{c}{5 Outliers} \\ \hline
Loc &   & Bias   & SD    & Coverage & Bias   & SD    & Coverage  & Bias   & SD    & Coverage \\ \hline \hline
Bayes
& 3     & -0.102 & 0.237 & 0.948    & -0.188 & 0.235 & 0.904     & -0.450 & 0.241 & 0.584 \\
80 & 5     & -0.101 & 0.237 & 0.950    & -0.188 & 0.235 & 0.907     & -0.451 & 0.241 & 0.582 \\
& 10    & -0.101 & 0.237 & 0.946    & -0.188 & 0.236 & 0.907     & -0.450 & 0.241 & 0.574 \\
\hline
Bayes
& 3     & -0.149 & 0.228 & 0.894    & -0.282 & 0.221 & 0.757     & -0.635 & 0.192 & 0.153 \\
99 & 5     & -0.151 & 0.231 & 0.893    & -0.290 & 0.228 & 0.754     & -0.704 & 0.231 & 0.147 \\
& 10    & -0.151 & 0.231 & 0.887    & -0.290 & 0.228 & 0.760     & -0.704 & 0.231 & 0.148 \\

\hline
Freq
& 3     & -0.087 & 0.238 & 0.922    & -0.172 & 0.236 & 0.872     & -0.429 & 0.242 & 0.454 \\
80 & 5     & -0.087 & 0.238 & 0.922    & -0.172 & 0.236 & 0.872     & -0.429 & 0.242 & 0.454 \\
& 10    & -0.087 & 0.238 & 0.922    & -0.172 & 0.236 & 0.872     & -0.429 & 0.242 & 0.454 \\

\hline
Freq
& 3     & -0.140 & 0.230 & 0.898    & -0.275 & 0.223 & 0.753     & -0.633 & 0.188 & 0.115 \\
99 & 5     & -0.140 & 0.231 & 0.897    & -0.279 & 0.229 & 0.751     & -0.693 & 0.231 & 0.115 \\
& 10    & -0.140 & 0.231 & 0.897    & -0.279 & 0.229 & 0.751     & -0.693 & 0.231 & 0.115 \\
\hline
Biweight
& 3     & -0.091 & 0.246 & 0.954    & -0.175 & 0.252 & 0.915     & -0.443 & 0.275 & 0.645 \\
& 5     & -0.018 & 0.237 & 0.974    & -0.019 & 0.236 & 0.972     & -0.022 & 0.243 & 0.967 \\
& 10    & -0.017 & 0.236 & 0.977    & -0.018 & 0.234 & 0.971     & -0.018 & 0.236 & 0.969 \\
\hline
\end{tabular}
\end{center}
\caption{Simulation results corresponding to Table \ref{simtable1} using Huberized and Tukey Biweight estimates
Bayes 80 and Bayes 99 correspond to replacing the log posterior with a Huber loss function using
cut points at the 0.8 and 0.99 critical values for a normal distribution. Freq 80 and Freq 99
correspond to obtaining estimators by minimizing the corresponding Huber loss functions.   In this
case, coverage corresponds to confidence intervals calculated as the estimator plus or minus
two estimated standard errors. Biweight correspond to replacing the log posterior with Tukey's biweight loss with cutpoint 4.685. } \label{hubsimtable1}
\end{table}

\begin{table}
\begin{center}
\begin{tabular}{|l|rrr|rrr|} \hline
\multicolumn{7}{|c|}{No Outliers} \\
\hline & \multicolumn{3}{|c|}{Shape} & \multicolumn{3}{|c|}{Scale} \\
                     & Bias   & Variance & Coverage & Bias   & Variance & Coverage\\ \hline
Posterior            & -0.088 & 0.571    & 0.9516 &  0.004 & 0.0005   & 0.9562 \\
Hellinger            & -0.081 & 1.005    & 0.850 &  0.010 & 0.0011   & 0.839 \\
Negative Exponential & -0.182 & 0.739    & 0.9436 &  0.008 & 0.0007   & 0.9556 \\
\hline \hline \multicolumn{7}{|c|}{Outlier at 20} \\
\hline & \multicolumn{3}{|c|}{Shape} & \multicolumn{3}{|c|}{Scale} \\
                     & Bias   & Variance & Coverage & Bias   & Variance & Coverage \\ \hline
Posterior            & -3.068 & 0.001    & 0 &  0.988 & 0.00006  & 0 \\
Hellinger            & -0.013 & 1.046    & 0.8508 & 0.010  & 0.0011   & 0.8440 \\
Negative Exponential & -0.210 & 0.725    & 0.9496 & 0.010  & 0.0008   & 0.9586 \\
\hline
\end{tabular}
\end{center}
\caption{A simulation study for an exponential-gamma using the usual posterior, the
Hellinger posterior and the Negative Exponential posterior. Columns give the
bias and variance of the posterior mean and the coverage of the central 95\%
credible interval based on 5000 simulations. Note that the same data sets are
used in for both tables, an outlier being added to the data sets when
calculating the quantities in the bottom table. The shape parameter is given in
the first column and the scale parameter in the second in each entry.  }
\label{simtable2}
\end{table}

\begin{figure}
\begin{center}
\includegraphics[height=8cm]{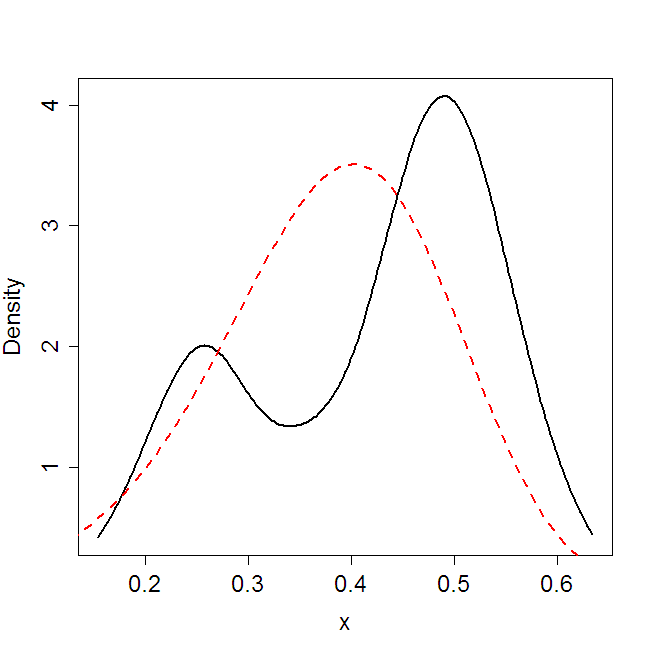}
\end{center}
\caption{Comparisons of $g_n(x)$ (dashed) and the true exponential-$\Gamma(5,0.25)$
density generating the data (solid). Hellinger distance estimators are
sensitive to valleys in the density that are due to a density estimate naively
applied to skewed data. Negative exponential disparity are robust towards these
effects as well.} \label{gamma.holes}
\end{figure}

\subsection{Linear Regression Models} \label{sim:linreg}

Figure \ref{linreg.fig} provides example D-posterior distributions of all
regression disparities described in Sections \ref{sec:conditional} and
\ref{sec:reduction} along with the posterior for an example 3-dimensional
linear regression based on 30 points. Both Hellinger and negative exponential
disparities were used. Covariates were generated from a standard normal
distribution with errors also generated from a standard normal distribution.
The likelihood is noticeably more concentrated than the disparity-based
posteriors, but all exhibit broadly similar shapes.

We have supplemented this experiment with a simulation. 1,000 data sets were
simulated from a linear regression model on three covariates. The covariates
were chosen from a three dimensional normal distribution with unit marginal
variances and correlation 0.5 between each pair of covariates. These were held
fixed across all simulations. The parameters in the model were chosen as
$(\beta_0,\beta_1,\beta_2,\beta_3,\sigma^2) = (1,1,1,1,1)$. Bandwidths were
chosen using the criterion in \citet{SheatherJones91} based on the observed
covariates and maximum likelihood estimates of the residuals. Both Hellinger
distance and negative exponential disparity were considered, and the
conditional formulation, marginal formulation and conditional-homoscedastic
form were used to estimate the five parameters in the model. Additionally, we
obtained Huberized estimators by replacing the log likelihood with a Huber loss
function and minimizing the Huber loss function. This resulted in
seven estimators including a Gaussian posterior. For each estimate a random-walk
Metropolis algorithm was run for 10,000 steps and EDAP estimates were
calculated from every fifth sample in the second half of the chain. The results
from this study are given in Table \ref{linreg.table}. Here we observe good
agreement between all disparity-based methods and the likelihood. The choice
between these will therefore depend on the amount of data available and the
dimension of the covariate space. Computing time is reported in this simulation
as well where disparity methods increase computing costs by a factor of between
3 and 5 while the use of a Huber loss within MCMC corresponds to about a 1.5
factor increase in cost. We note, however, that relative computing costs in
interpreted languages such as \texttt{R} are strongly dependent on the details
of implementation.

\begin{figure}
\begin{center}
\begin{tabular}{cc}
\includegraphics[height=8cm]{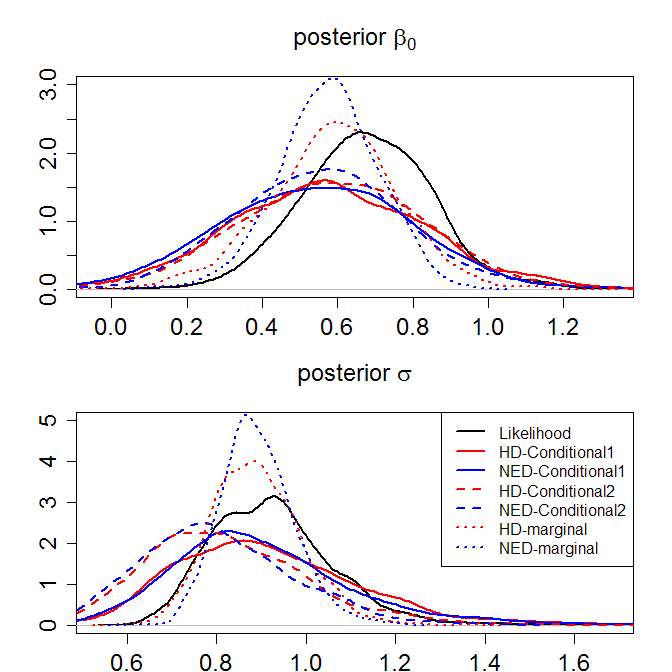} &
\includegraphics[height=8cm]{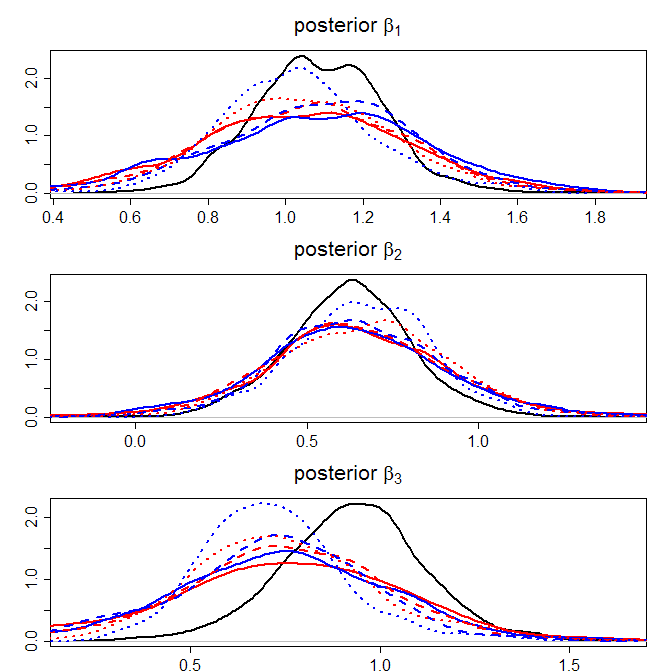}
\end{tabular}
\end{center}
\caption{D-posterior inference and linear regression. Top left: posterior for
$\beta_0$, bottom left: posterior for $\sigma$. Right: posterior for each
$\beta_i$, $i=1,\ldots,3$. Thick lines: posterior based on likelihood. Solid
lines: based on $g_n^{(c)}$, dashed: based on $g_n^{(c2)}$, dotted: marginal
formulation. Black: Hellinger distance, grey: negative exponential disparity.
Here 'Conditional 1' indicates the conditional density estimate in Section
\ref{sec:conditional}, 'Conditional 2' refers to the conditional-homoscedastic
approach.} \label{linreg.fig}
\end{figure}

\begin{table}
\begin{center}
\begin{tabular}{|l|c|c|c|c|c|c|} \hline
            &  Sigma        & beta 0         & beta 1         & beta 2         & beta 3 & time \\ \hline
Likelihood  & 1.030 (0.137) & 0.996  (0.186) & 0.993  (0.194) & 0.998  (0.217) & 1.006  (0.301)  & 20.251 \\
HD          & 0.974 (0.154) & 1.000  (0.213) & 0.993  (0.204) & 0.994  (0.239) & 0.980  (0.301)  & 62.315 \\
NED      & 0.926 (0.136) & 0.997  (0.201) & 0.994  (0.206) & 1.001  (0.224) & 0.964  (0.294)  & 70.812 \\
\hline
HD hom      & 0.881 (0.181) & 0.995  (0.338) & 0.987  (0.198) & 0.995  (0.217) & 0.980  (0.296)  & 62.184 \\
NED hom  & 0.829 (0.157) & 0.993  (0.235) & 0.988  (0.211) & 0.996  (0.228) & 0.974  (0.313)  & 70.649 \\
\hline
HD marg     & 1.062 (0.133) & 0.997  (0.193) & 0.994  (0.203) & 0.997  (0.225) & 1.003  (0.309)  & 102.959 \\
NED marg & 1.087 (0.137) & 1.000  (0.199) & 0.994  (0.210) & 0.998  (0.232) & 1.005  (0.321)  & 102.966 \\
\hline
Huber EAP & 0.878 (0.123)  & 0.998  (0.192)  & 0.994  (0.198)  & 1.000  (0.221)  & 1.005  (0.306) & 34.861 \\
Huber min & 0.919 (0.207)  & 0.999  (0.191)  & 0.995  (0.198)  & 1.000  (0.222)  & 1.006  (0.307) & 0.013 \\
\hline
\end{tabular}
\end{center}
\caption{Simulation results for linear regression. Columns give mean and
standard deviation of EDAP estimates of parameters  and CPU time based on 1,000 simulated
data sets. Rows correspond to posteriors using Likelihood, using conditional
density estimates with Hellinger distance (HD) and negative exponential disparity
(NED), using conditional-homoscedastic density estimates with Hellinger
distance (HD hom) and negative exponential disparity (NED hom), using the marginal
formulation of Hellinger distance (HD marg) and negative exponential disparity (NED marg) and
using a Huber loss scaled by $\sigma$ in place of the log likelihood (Huber EAP) and minimizing the Huber loss (Huber min).
In both the later case the cut point was chosen at the 0.8 critical value of a standard normal.}
\label{linreg.table}
\end{table}

\subsection{One-Way Random Effects Models} \label{sim:randeffect}

Figure \ref{randeffect} demonstrates the differences resulting from
robustifying different distributional assumptions in the one-way random effects
model described in Section \ref{sec:plugin}. We simulated a set of five groups
of twenty observations. Each group had variance 0.2 with a mean drawn from a
$N(\mu,1)$ population. This produces a random-effects model and the goal is to
estimate $\mu$. The plots in Figure \ref{randeffect} show that the use of
disparity methods in either the random effect or on the residual process or on
both provide very similar distributions to the correct posterior. When an
outlying group is added with mean 40, those methods that replace the random
effects distribution with a disparity are unaffected while those that do not
are substantially biased.

\begin{figure}
\begin{center}
\begin{tabular}{cc}
\includegraphics[height=5cm]{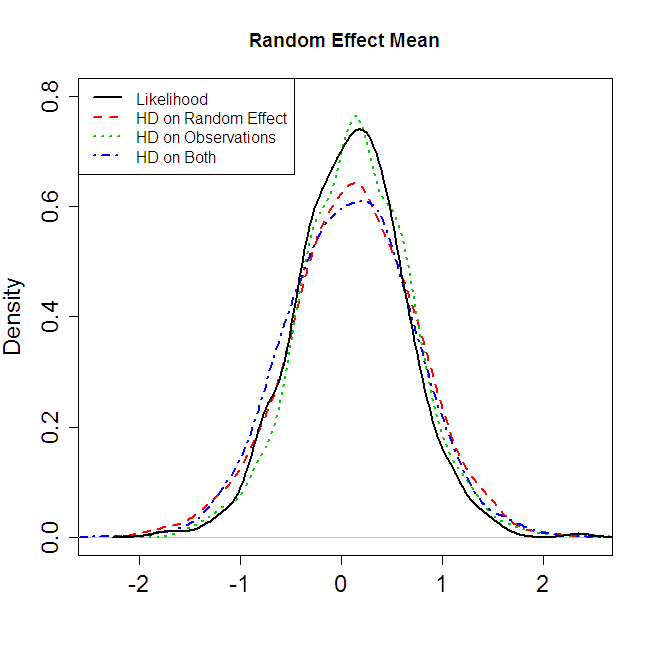}
&
\includegraphics[height=5cm]{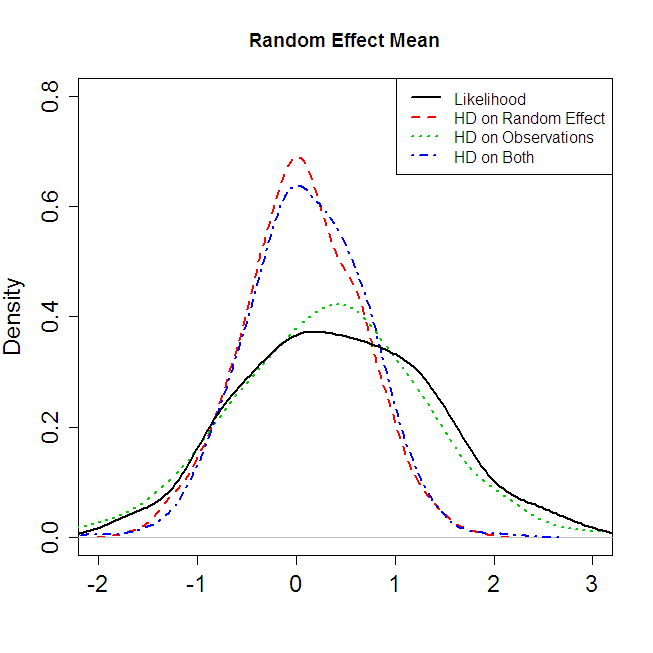}
\\
\includegraphics[height=5cm]{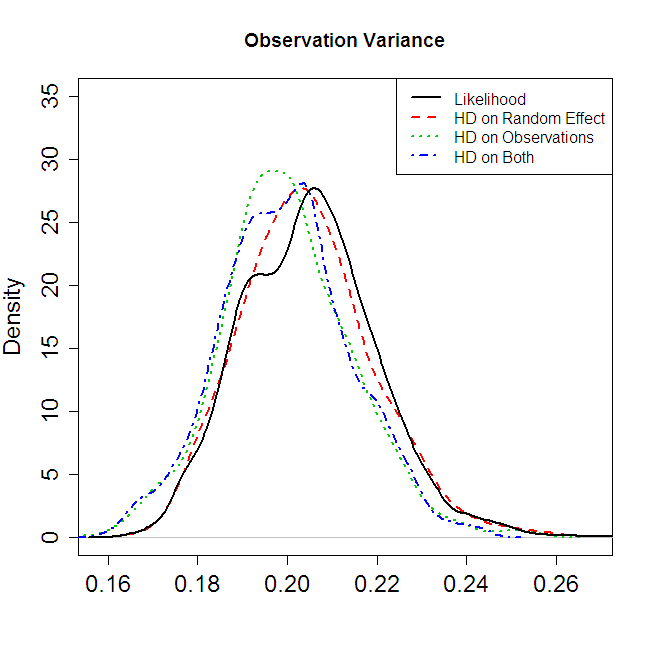}
&
\includegraphics[height=5cm]{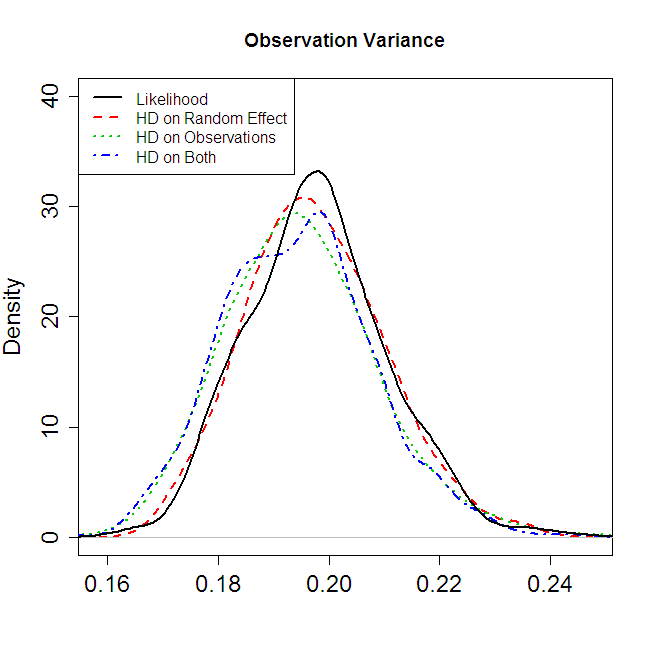}
\\
\includegraphics[height=5cm]{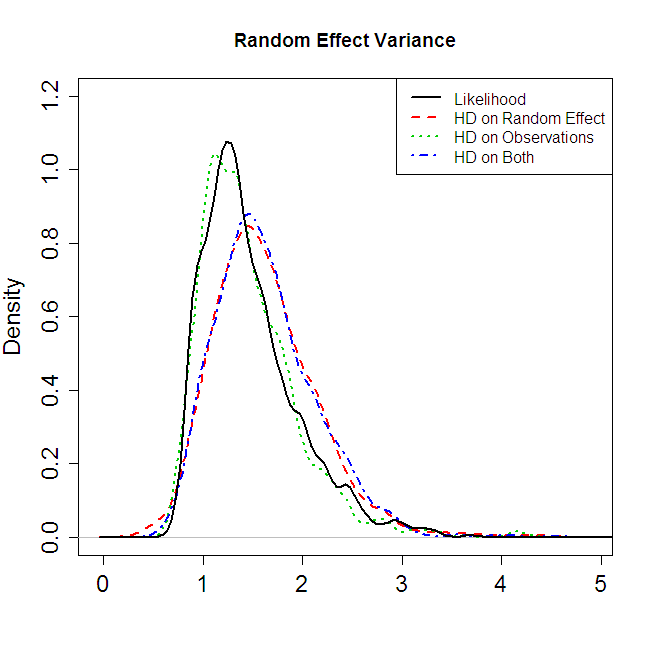}
&
\includegraphics[height=5cm]{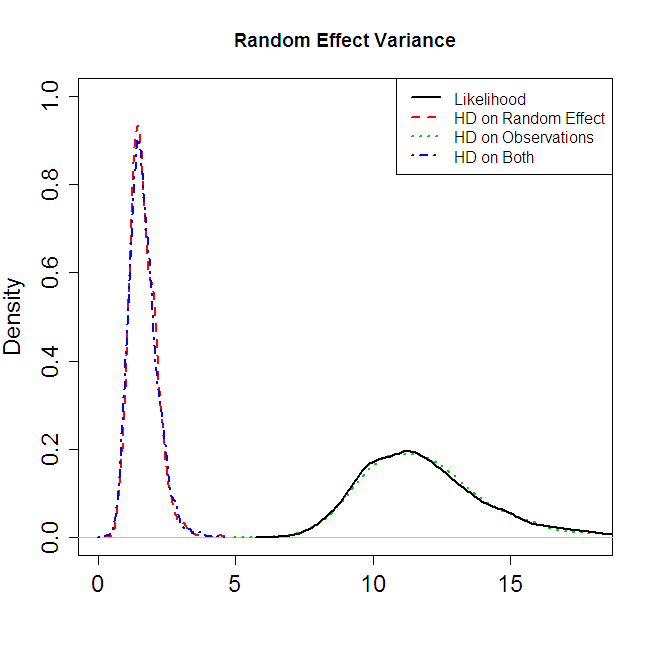}
\\
\end{tabular}
\end{center}
\caption{Posterior distributions for a random-effects model. 5 groups of 20
observations each with mean drawn from $N(\mu,1)$ and variance 0.2. $\mu$ was
given a $N(0,1)$ prior. Posterior densities for $\mu$ were estimated from every
50th observation in the last half of a random-walk Markov chain run for 100,000
steps. Hellinger distance was used to replace the component of the likelihood
representing $Y_{ij}-X_i$, representing $X_i$ and both. Left: posterior
densities for $\mu$ from the original data when the model is correct. Right:
posterior densities after adding a further group of 5 observation generated
from $N(40,0.2)$. First row: estimates for $\mu$, second: $\sigma^2$, third:
$\tau^2$.} \label{randeffect}
\end{figure}

In order to verify the apparent success of this method, we conducted a
simulation study of a one-way random effects model with ten random effects and
five observations per random effect. The random effects were simulated from a
standard normal distribution, while the observations were Gaussian, centered on
the random effect and with standard deviation 0.2.  We simulated 1,000 versions
of these data. For each version a random-walk Metropolis scheme was run to
sample from the posterior, the posterior with the observation likelihood
replaced by H-posterior, the posterior with the random effect likelihood
replaced by H-posterior and the posterior with both replacements. All MCMC
schemes were run for 10,000 steps with posterior distributions calculated based
on every 5th sample from the second half of the chain. We additionally added a
further random effect with five observations distributed around the value 40
with standard deviation 0.2.  Bandwidth parameters were chosen by the selection
criterion of \citet{SheatherJones91} based on maximum-likelihood estimates of
random effects and residual errors. The results of this simulation are
summarized in Table \ref{randeffect.sim}. Here we see that the estimation of
$\sigma$ is biased downwards by the estimation of a conditional density for
each random effect, based on only five observations and there is more
uncertainty in the estimate of $\tau$ when Hellinger distance is used in place
of the random effect log likelihood. The disparity-based methods otherwise
perform very similarly to the true likelihood. When an outlying random effect
is added, replacing the random effect likelihood with Hellinger distance
robustified inference, where those posteriors without this replacement were
severely biased.

\begin{table}
\begin{center}
\begin{tabular}{|l|rrr|rrr|rrr|} \hline
\multicolumn{10}{|c|}{No Outliers} \\ \hline
           &  $\mu$ &   sd($\mu$) & coverage &  $\sigma$ & sd($\sigma$) & coverage & $\tau$ & sd($\tau$) & coverage
           \\ \hline
Likelihood & 0.0013 & 0.286 & 0.947 & 0.200 & 0.0230 & 0.940 & 0.979 & 0.234 & 0.919 \\
HD - obs   & 0.0035 & 0.285 & 0.937 & 0.200 & 0.0231 & 0.939 & 1.070 & 0.342 & 0.900 \\
HD - rand  & 0.0021 & 0.286 & 0.936 & 0.191 & 0.0259 & 0.686 & 0.982 & 0.231 & 0.928 \\
HD - both  & 0.0024 & 0.291 & 0.933 & 0.191 & 0.0258 & 0.692 & 1.068 & 0.342 & 0.899 \\
\hline \multicolumn{10}{|c|}{Outlying Random Effect} \\ \hline
           &  $\mu$ &   sd($\mu$) & coverage &  $\sigma$ & sd($\sigma$) & coverage & $\tau$ & sd($\tau$) & coverage
           \\ \hline
Likelihood & 0.365 & 0.197 & 1.000 & 0.200 & 0.0222 & 0.930 & 10.11 & 0.163 & 0.000 \\
HD - obs   & 0.002 & 0.288 & 0.926 & 0.200 & 0.0223 & 0.922 & 1.088 & 0.389 & 0.887 \\
HD - rand  & 0.353 & 0.252 & 1.000 & 0.191 & 0.0249 & 0.684 & 10.13 & 0.196 & 0.000 \\
HD - both  & 0.002 & 0.295 & 0.917 & 0.191 & 0.0249 & 0.674 & 1.066 & 0.325 & 0.885 \\
\hline
\end{tabular}
\end{center}
\caption{ A simulation study for a one-way random effect model from using the
posterior, replacing the observation likelihood with a conditional Hellinger
distance, replacing the random effect likelihood with Hellinger distance and
making both replacements. The columns give the mean, standard deviation and
coverage of $\mu$, $\sigma^2$ and $\tau^2$ based on 1,000 simulations. The
lower table indicates the effect of adding an outlying random effect at 40. A
total of 8 simulations were excluded due to poor convergence of the MCMC
chain.} \label{randeffect.sim}
\end{table}

\section{Data} \label{sec:data}

Table \ref{parasite.data} provides the values of the parasitology data set used
in Section \ref{sec:parasite}. The data used for the class survey data in
Section \ref{sec:survey} are provided in Table \ref{survey.data}; they are
graphed in Figure \ref{surveydata.fig}.

\begin{table}
\begin{center}
\begin{tabular}{|l|ccccccc|} \hline
Horse & 1 & 2 & 3 & 4 & 5 & 6 & 7 \\ \hline
Pre-treatment & 2440 & 1000 & 1900 & 1820 & 3260 & 300 & 660 \\
Post-treatment & 580 & 320  & 400  & 160  & 60   & 40  & 120 \\ \hline
\end{tabular}
\end{center}
\caption{ Data used in Section \ref{sec:parasite}: pre- and post-treatment
fecal egg count for seven horses on one farm.} \label{parasite.data}
\end{table}

\begin{table}
\begin{center}
\begin{tabular}{|ccccc|} \hline
Status & 35 & 45 & 55 & 65 \\ \hline
a & 11.51293 & 11.775290 & 11.98293 & 12.10071 \\
a & 11.91839 & 12.206073 & 12.61154 & 12.20607 \\
a & 11.69525 & 12.100712 & 12.42922 & 12.42922 \\
a & 11.69525 & 12.100712 & 12.38839 & 12.79386 \\
a & 10.71442 & 10.819778 & 10.81978 & 10.81978 \\
a & 11.15625 & 11.512925 & 11.84940 & 11.98293 \\
a & 10.51867 & 10.596635 & 10.71442 & 11.00210 \\
a & 12.20607 & 12.206073 & 12.20607 & 12.20607 \\
a & 9.21034  &  9.903488 & 10.12663 & 10.30895 \\
a & 10.59663 & 11.002100 & 11.28978 & 11.40756 \\
a & 14.22098 & 12.100712 & 14.91412 & 15.60727 \\
f & 11.51293 & 11.918391 & 11.51293 & 10.81978 \\
f & 11.00210 & 11.002100 & 11.15625 & 11.15625 \\
f & 11.91839 & 12.206073 & 12.42922 & 12.42922 \\
f & 11.40756 & 11.695247 & 11.73607 & 11.77529 \\
f & 11.91839 & 13.122363 & 13.30468 & 13.30468 \\
f & 11.51293 & 11.695247 & 11.91839 & 12.20607 \\
f & 11.15625 & 11.512925 & 11.51293 & 11.69525 \\
f & 12.20607 & 12.611538 & 12.76569 & 11.00210 \\
f & 10.81978 & 11.002100 & 11.00210 & 11.00210 \\
f & 10.46310 & 11.002100 & 11.08214 & 11.08214 \\ \hline
\end{tabular}
\end{center}
\caption{ Class survey data used in Section \ref{sec:survey}; columns give
American (a) or foreign (f) status, and log expected salary at ages 35, 45, 55,
and 65.} \label{survey.data}
\end{table}

\begin{figure}
\begin{center}
\includegraphics[height=10cm]{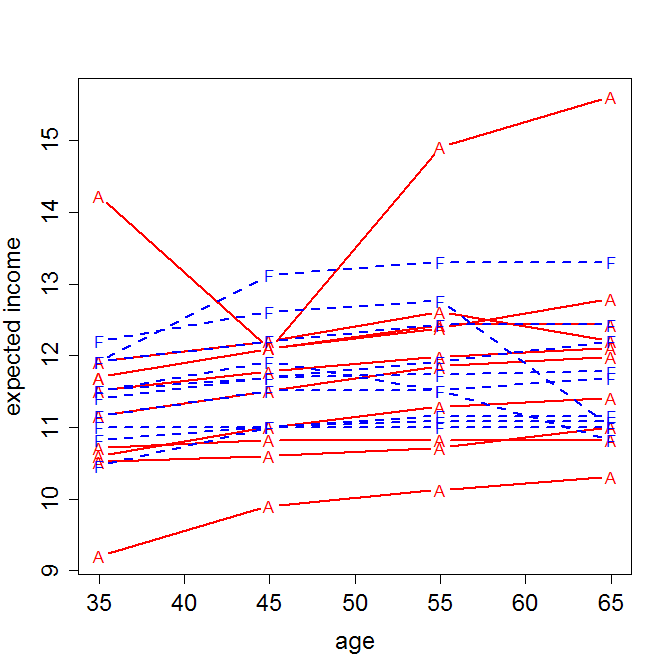}
\end{center}
\caption{Responses to an in-class survey on expected income at ages 35, 45, 55
and 65. Students were either foreign born (dashed) or American (solid).}
\label{surveydata.fig}
\end{figure}

\subsection{Random Slope Models for Class Survey Data} \label{sec:randomslope}

A further model exploration  of the class survey data allows a random slope for each student in addition
to the random offset. The model now becomes
\begin{equation} \label{randslope}
Y_{ijk} = b_{0ij} + b_{1ij} t_k + \epsilon_{ijk}
\end{equation}
with additional distributional assumptions
\[
b_{1ij} \sim N(\beta_{1j},\tau_1^2)
\]
and an additional term
\[
-\frac{1}{2 \tau_1^2} \sum_{i=1}^n \sum_{j \in \{a,f\}}
(b_{1ij}-\beta_{1j})^2
\]
added to \eqref{randintercept.cdll}. Here, this term can be robustified in a
similar manner to the robustification of the $b_{0ij}$. However, we note
that a robustification of the error terms would require the estimation of a
conditional density for each $ij$ -- based on only four data points.  We viewed
this as being too little to achieve reasonable results and therefore employed
the marginal formulation described in Section \ref{sec:reduction}.
Specifically, we first obtained residuals $e_{ijk}$ for the random slope model
from the maximum likelihood estimates for each subject-specific effect and
estimated
\[
\hat{\sigma}^2 = \frac{1}{0.674 \sqrt{2}} | e_{ijk} - \mbox{median}(e) |.
\]
Following this, we estimated a combined density for all residuals, conditional
on the random effects
\[
g_n^{(m)}(e;\betabold) = \frac{1}{8n c_n} \sum_{i=1}^n \sum_{j \in \{a,f\}}
\sum_{k=1}^4 K\left( \frac{e - (Y_{ijk}-b_{0ij} -b_{1ij}t_k)}{c_n}
\right)
\]
and replaced the first term in \eqref{randintercept.cdll} with
$-8nD(g_n^{(m)}(\cdot;\betabold),\phi_{0,\hat{\sigma}^2}(\cdot))$. Following the
estimation of all other parameters, we obtained new residuals $\tilde{e}_{ijk} =
Y_{ijk} - \tilde{b}_{0ij} - b_{1ij} t_k$
where the $\tilde{b}_{0ij}$ and $\tilde{b}_{1ij}$ are the EDAP
estimators. We then re estimated $\sigma^2$ based on its H-posterior using the
$\tilde{e}_{ijk}$ as data. In this particular case a large number of outliers
from a concentrated peak (see Figure \ref{randomslope.fig}) meant that
the use of Gauss-Hermite quadrature in the evaluation of
\[
HD(g_n^{(m)}(\cdot,\tilde{\betabold}),\phi_{0,\sigma^2}) = 2 - 2\int \left(
\sqrt{g_n^{(m)}(e;\tilde{\betabold})}/\sqrt{\phi_{0,\sigma^2}(e)} \right)
\phi_{0,\sigma^2}(e)  de
\]
suffered from large numerical errors and we therefore employed a Monte Carlo
integral based on 400 data points drawn from $g_n^{(m)}$ instead, using the
estimate \eqref{mcestimate}. To estimate both $\sigma^2$ and the other
parameters we used a Metropolis random walk algorithm which was again run for
200,000 iterations with estimates based on every 100th sample in the second
half of the chain.

Some results from this analysis are displayed in Figure \ref{randomslope.fig}.
The residual distribution of the $\tilde{e}_{ijk}$ show a very strong peak and
a number of isolated outliers. The estimated standard deviation of the residual
distribution is therefore very different between those methods that are robust
to outliers and those that are not; the mean posterior $\sigma$ was increased
by a factor of four between those methods using a Hellinger disparity and those
using the random effect log likelihood. The random slope variance was estimated
to be small by all methods --  we speculate that the distinction between random
effect log likelihoods and Hellinger methods is bias due to bandwidth size --
but this was not enough to overcome the differences between the methods
concerning the distinction between $\beta_{f1}$ and $\beta_{a1}$.

\begin{figure}
\begin{center}
\begin{tabular}{cc}
\includegraphics[height=6.5cm,angle=0]{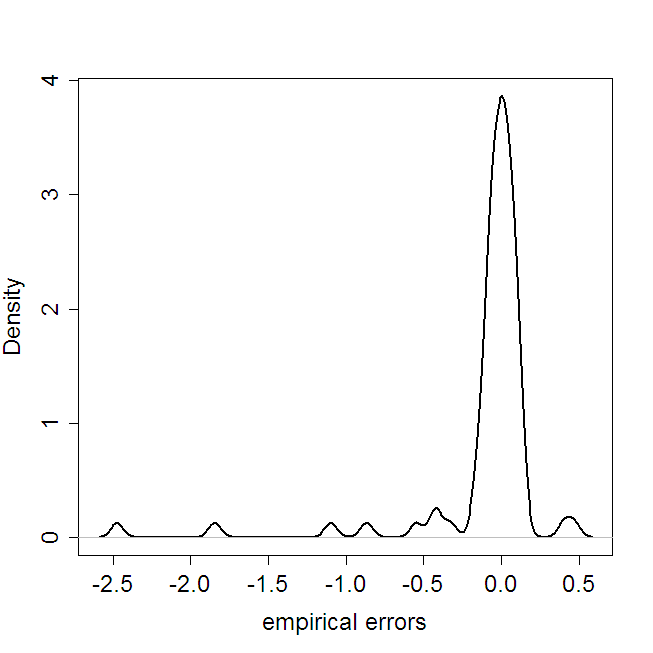} &
\includegraphics[height=6.5cm,angle=0]{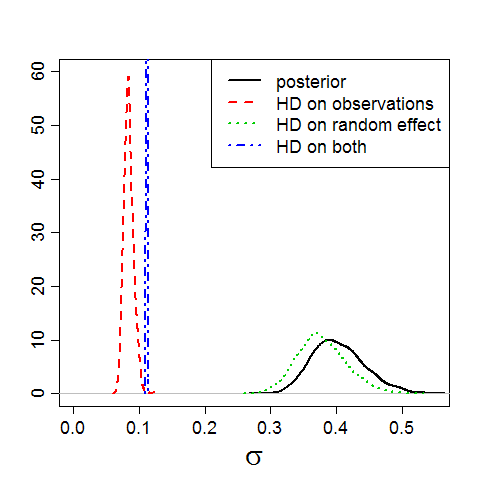} \\
\includegraphics[height=6.5cm,angle=0]{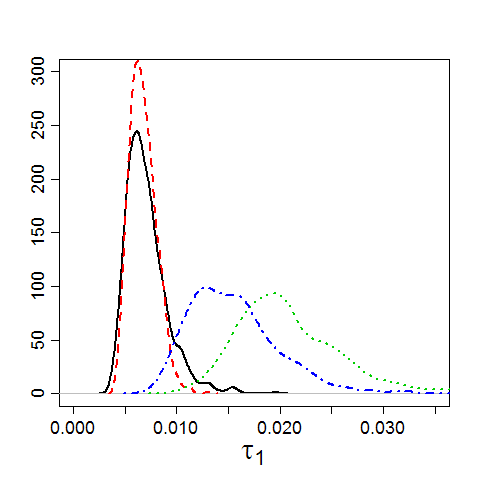} &
\includegraphics[height=6.5cm,angle=0]{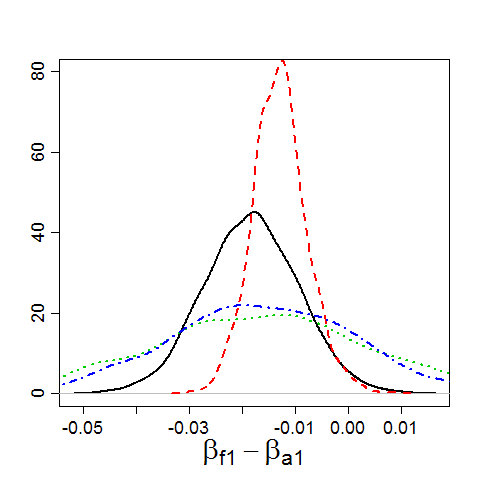}
\end{tabular}
\end{center}
\caption{Analysis of a random-slope random-intercept model for the class survey data.
Top left: a density estimate of the errors following a two-step procedure with the error variance
held constant. This shows numerous isolated outliers than create an ill-conditioned problem for
Gauss-Hermite quadrature methods. Top right: estimates of residual standard deviation replacing
various terms in the log likelihood with Hellinger distance. The effect of outliers is clearly apparent
in producing an over-estimate of variance. Bottom left: estimated variance of the random slope. Bottom
right: estimated difference in mean slope between American and foreign students.} \label{randomslope.fig}
\end{figure}

\end{document}